\documentclass[12pt,a4paper]{article}

\usepackage[margin=1in]{geometry}
\usepackage{amsmath}
\usepackage{amsthm}
\usepackage{amsfonts}
\usepackage{amssymb}
\usepackage{graphicx}
\usepackage{xcolor}
\usepackage[hidelinks]{hyperref}
\usepackage{cleveref}
\usepackage{subfigure}
\usepackage{caption}
\usepackage{algorithm}
\usepackage{algorithmicx}
\usepackage{algpseudocode}
\usepackage{multirow,booktabs,bigstrut}
\usepackage{adjustbox}
\usepackage{indentfirst}

\makeatletter
\let\grfm@legacysubfigure\subfigure
\renewcommand{\subfigure}{%
  \def\grfm@subfigureenvironment{subfigure}%
  \ifx\@currenvir\grfm@subfigureenvironment
    \expandafter\grfm@subfigurebegin
  \else
    \expandafter\grfm@legacysubfigure
  \fi}
\newcommand{\grfm@subfigurebegin}[2][b]{%
  \minipage[#1]{#2}%
  \captionsetup{type*=subfigure}}

\makeatother

\newtheorem{theorem}{Theorem}
\theoremstyle{remark}
\newtheorem{remark}{Remark}

\title{GR-FM: Geometrically Regularized Flow Matching for SDF-Based Medical Image Segmentation}

\author{
Yuxin Ai\thanks{School of Mathematics, Harbin Institute of Technology,
Harbin, China
(\href{mailto:24B312010@stu.hit.edu.cn}{\texttt{24B312010@stu.hit.edu.cn}},
\href{mailto:mathgzc@hit.edu.cn}{\texttt{mathgzc@hit.edu.cn}},
and \href{mailto:zhangdazhi@hit.edu.cn}{\texttt{zhangdazhi@hit.edu.cn}}).}
\and Zhichang Guo\footnotemark[1]
\and Fanghui Song\footnotemark[1]\thanks{Corresponding author
(\href{mailto:fanghuisong@stu.hit.edu.cn}{\texttt{fanghuisong@stu.hit.edu.cn}}).}
\and Dazhi Zhang\footnotemark[1]
}

\hypersetup{
  pdftitle={GR-FM: Geometrically Regularized Flow Matching for SDF-Based Medical Image Segmentation},
  pdfauthor={Y. Ai, Z. Guo, F. Song, and D. Zhang}
}

\begin{document}

\maketitle

\begin{abstract}
Medical image segmentation remains challenging in terms of accurate boundary localization and complex-structure preservation, as target regions often exhibit weak boundaries, fine-grained structures, and irregular shapes, while high-quality images and precise annotations are usually limited. Existing generative segmentation approaches based on Flow Matching mainly learn velocity fields in the state space but lack explicit constraints on the spatial regularity and distance-field properties of the recovered representation. This limitation may lead to spatial oscillations, boundary displacement, and discontinuities in fine structures. To address these issues, we propose an image-conditioned Geometrically Regularized Flow Matching framework, termed GR-FM, which formulates segmentation as a continuous probability transport process from an initial distribution to a target implicit representation distribution. Instead of directly modeling binary masks, GR-FM adopts the SDF to describe the target structure, where each pixel is represented by its signed distance to the object boundary, resulting in a continuous geometric field. Flow Matching and ordinary differential equations are then employed to achieve deterministic and efficient distribution transport. The training objective further incorporates a biharmonic regularization term and an Eikonal constraint to enhance the spatial smoothness, structural consistency, and distance-field characteristics of the recovered representation. Moreover, we analyze the continuous transport process under geometric constraints and investigate the evolution of the zero level set. Experiments conducted on the MoNuSeg, GlaS, and DRIVE datasets demonstrate that GR-FM achieves competitive performance in both region overlap and boundary accuracy, reduces performance variation, and maintains stable segmentation results with only a small number of integration steps.
\end{abstract}

\section{Introduction}

Medical image segmentation aims to identify organs and lesions at the pixel level from different imaging modalities. It provides essential support for clinical diagnosis, treatment planning, and postoperative assessment~\cite{azad2024medical, gao2025medical, xu2024advances}. It has been widely applied to tumor delineation~\cite{sarkar2025deep}, vessel extraction~\cite{cai2013vessel}, organ localization, and lesion analysis~\cite{wasserthal2023totalsegmentator}. However, target regions in medical images often have weak boundaries, fine structures, and highly irregular shapes. Medical imaging also faces variations in image quality and a limited number of high-quality images. These factors make it difficult to locate unclear boundaries, preserve the continuity of thin structures, and recover complex anatomical shapes. They can also reduce model generalization and structural stability across different imaging conditions. Medical image segmentation therefore requires not only high pixel-wise accuracy but also reliable geometric representation to preserve stable boundaries and anatomical structures.

Signed distance function (SDF) represents a target region as a continuous scalar field and defines its boundary by the zero level set. It transforms discrete mask prediction into the recovery of a continuous implicit geometric representation. Compared with binary masks, SDFs encode boundary locations, inside--outside relationships, and spatial distance information. They are therefore well suited to representing weak boundaries, fine structures, and irregular shapes. Flow Matching~\cite{lipman2022flow} learns a continuous-time velocity field that transports a simple prior distribution to a target data distribution. It formulates segmentation as a deterministic transport process governed by an ordinary differential equation. Compared with diffusion models based on stochastic reverse denoising, continuous flow models can provide stable inference with fewer sampling steps and allow structural priors to be included in the training objective. FlowSDF~\cite{bogensperger2025flowsdf} further combines SDFs with conditional continuous flows, providing a natural framework that connects continuous generative modeling with implicit geometric representation.

Although FlowSDF introduces SDF representations into continuous flow modeling and learns conditional transport from a prior distribution to a target SDF distribution, its optimization objective is still centered on velocity matching in the state space. It does not explicitly constrain the spatial derivatives of the generated SDF. Minimizing the velocity regression error only encourages the model to approximate the reference transport direction. It does not directly ensure that the terminal SDF has sufficient higher-order spatial regularity or satisfies the Eikonal property of an ideal distance field. The generated field may therefore contain distorted spatial gradients and local high-frequency oscillations, which can cause boundary shifts, irregular contours, and breaks in fine structures. Explicitly incorporating higher-order spatial regularity and distance-field properties into the training objective is thus important for improving SDF recovery and the geometric reliability of segmentation while preserving the efficient transport mechanism of Flow Matching.

We propose a higher-order geometry-constrained continuous probability flow framework that incorporates geometric priors into flow-field learning. The proposed framework significantly improves the stability and geometric fidelity of the transport process while retaining the efficiency of flow-based segmentation. The main contributions of this work are summarized as follows:

\begin{itemize}

\item 
We introduce Geometrically Regularized Flow Matching (GR-FM), a geometry-aware continuous transport formulation for SDF-based medical image segmentation. The proposed framework integrates higher-order geometric regularization into the flow matching dynamics, providing explicit control over boundary stability, shape preservation, and topological consistency during SDF evolution.

\item 
We adopt SDFs as the target representation for segmentation, transforming discrete mask prediction into the transport of continuous implicit geometric representations. By introducing an Eikonal constraint, the proposed framework preserves the distance-function property throughout SDF evolution, thereby improving the geometric interpretability of intermediate states and enhancing the boundary localization stability of final segmentation results.

\item 
We provide a theoretical analysis of the continuous transport process under high-order geometric constraints. Specifically, we characterize the transport properties of zero level sets during continuous flow evolution, thereby providing theoretical support for the proposed framework.
\end{itemize}
\section{Related Work}

\subsection{SDF-based Segmentation}
\label{sec:sdf}

Let the image domain be $\Omega\subset\mathbb{R}^2$, the ground-truth mask $y:\Omega\to\{0,1\}$ denote the binary label, the region to be segmented $\Omega_{\mathrm{in}}\subset\Omega$, and its boundary $\Gamma=\partial\Omega_{\mathrm{in}}$. An SDF embeds the discrete mask into a continuous real-valued space. The signed distance $d:\Omega\to\mathbb{R}$ is defined as
\begin{equation}
    d(\mathbf{x}) =
    \begin{cases}
        -\operatorname{dist}(\mathbf{x},\Gamma), & \mathbf{x}\in\Omega_{\mathrm{in}}, \\[4pt]
        \operatorname{dist}(\mathbf{x},\Gamma),  & \mathbf{x}\notin\Omega_{\mathrm{in}},
    \end{cases}
    \label{eq:sdf-def}
\end{equation}
where $\operatorname{dist}(\mathbf{x},\Gamma)$ denotes the Euclidean distance from point $\mathbf{x}$ to the boundary $\Gamma$. By definition, the zero level set $\{\mathbf{x}\in\Omega:d(\mathbf{x})=0\}$ coincides exactly with the target boundary $\Gamma$, while $d(\mathbf{x})<0$ and $d(\mathbf{x})>0$ correspond to the interior and exterior of region $\Omega_{\mathrm{in}}$, respectively. Applying Eq.~\eqref{eq:sdf-def} pixel-wise to mask $y$ yields its signed distance representation.

To focus the model more on geometric information near the boundary neighborhood $\Gamma$, the distance values can be truncated at a threshold $\delta>0$. The truncated SDF is denoted by $s$, and its value at pixel location $\mathbf{x}_{ij}$ is defined as
\begin{equation}
    s(\mathbf{x}_{ij}) =
    \begin{cases}
        -\min\Bigl\{d(\mathbf{x}_{ij}),\;\delta\Bigr\},
        & \mathbf{x}_{ij}\in\Omega_{\mathrm{in}}, \\[6pt]
        \min\Bigl\{d(\mathbf{x}_{ij}),\;\delta\Bigr\},
        & \mathbf{x}_{ij}\in\Omega\setminus\Omega_{\mathrm{in}}, \\[6pt]
        0, & \mathbf{x}_{ij}\in\Gamma,
    \end{cases}
    \label{eq:sdf-trunc}
\end{equation}
where $\mathbf{x}_{ij}$ denotes the spatial coordinates of pixel $(i,j)$. The truncation operation attenuates the influence of large distance values in regions far from the boundary on model training, thereby guiding the model to focus more on local geometric structures near the boundary.

At inference, the model obtains the predicted SDF map $s_{\mathrm{pred}}$ through ODE numerical integration. Since the final segmentation results are evaluated and visualized in the form of binary masks, the continuous SDF representation is converted into discrete labels using a threshold $\tau$. Specifically, the predicted mask is obtained as

\begin{equation}
    y_{\mathrm{pred}}(\mathbf{x}_{ij})
    =
    \mathbf{1}_{\{s_{\mathrm{pred}}(\mathbf{x}_{ij})\leq \tau\}}.
    \label{eq:binary-infer-thresh}
\end{equation}

where $\tau$ denotes the decision threshold for binarization. In practice, $\tau$ can be selected within an appropriate range according to the characteristics of the target SDF. Through this transformation, the continuous SDF predicted by the model is converted into the final binary segmentation mask for subsequent evaluation and visualization.

\subsection{Continuous Transport Modeling for Medical Image Segmentation}
\label{sec:preliminaries}

Recent studies have adapted continuous flow models to dense prediction tasks. U-Net~\cite{ronneberger2015u} established a widely used encoder--decoder architecture for medical image segmentation. It recovers spatial details through hierarchical feature extraction and skip connections. Subsequent U-Net variants introduced residual connections, nested skip pathways, and multiscale feature fusion. These designs further improved local structure representation~\cite{huang2020unet,ibtehaz2020multiresunet,li2019residual,zhou2018unet++}. Attention mechanisms and Transformers were later integrated into U-Net-based architectures. They improve target localization and capture long-range dependencies~\cite{cao2022swin,chen2021transunet,liu2021swin,oktay2018attention}. In parallel, the DeepLab series developed another influential segmentation framework. It uses atrous convolution and atrous spatial pyramid pooling to capture multiscale context while preserving spatial resolution~\cite{chen2014semantic,chen2017deeplab,chen2017rethinking}. More recently, generative models have also been introduced into medical image segmentation. They provide new ways to model target distributions and the formation of segmentation structures.

Regularization complements architectural design by improving generalization and spatial reliability. Santos~\cite{santos2022avoiding} grouped regularization methods for convolutional neural networks into data augmentation, internal network modifications, and label regularization. General overfitting control, however, does not explicitly model spatial relations between pixels. Liu, Wang, and Tai~\cite{liu2022deep} incorporated spatial smoothness, region-volume, and star-shape priors into the softmax classification layer. Zhang and Guo~\cite{zhang2023new} proposed directional connectivity-enhancing regularization for fine structures in low-contrast images. Their method improves the continuity of cracks and retinal vessels, although the gain in sensitivity may be accompanied by a slight reduction in specificity. Higher-order geometric regularization provides another way to control spatial structure. Curvature, Laplacian, and biharmonic-type terms can suppress local oscillations and improve boundary smoothness and structural continuity~\cite{chen2020learning,droske2004level,lenz2023complementary}.

This deterministic formulation can be expressed in a common mathematical form. A two-dimensional medical image is represented as $\mathbf{I}\in\mathbb{R}^{M\times N\times 2}$, and its label is represented as $y\in\{0,1,\ldots,C-1\}^{M\times N}$, where $C$ denotes the number of classes. Given a training set $\mathcal{D}=\{(\mathbf{I}^{(n)},y^{(n)})\}_{n=1}^{N}$, a segmentation network learns parameters $\theta$ such that the prediction $\hat{y}=f_{\theta}(\mathbf{I})$ approaches the ground-truth annotation $y$. The optimization objective is written as
\[
\min_{\theta}
\mathbb{E}_{(\mathbf{I},y)\sim\mathcal{D}}
\left[
\mathcal{L}_{\mathrm{seg}}\bigl(f_{\theta}(\mathbf{I}),y\bigr)
\right].
\]

These end-to-end methods usually produce the final segmentation mask through a single forward mapping. They learn a direct relation between the input image and the target region. However, they do not explicitly describe how the target structure is formed from an initial state. This limitation motivates generative formulations that model both the conditional target distribution and the formation process.

Among these formulations, generative adversarial networks~\cite{haoqi2020cgan, xue2018segan, zhang2022comparative} learn the joint distribution of images and segmentation masks to improve structural consistency. Diffusion models~\cite{ho2020denoising, liu2025diffusion} represent mask generation as a gradual reverse denoising process and refine segmentation structures through multiple iterations. MedSegDiff~\cite{wu2024medsegdiff} applies conditional diffusion to medical image segmentation. Later methods combine Transformer and U-Net features~\cite{ajith2025transunet, chen2024hidiff, xing2023diff} to improve the modeling of global relations and multiscale structures. Although diffusion models can represent complex target distributions, their stochastic initialization and multistep reverse sampling increase inference cost and may introduce prediction uncertainty.

Flow-based generative models~\cite{kobyzev2020normalizing, lipman2022flow} use ordinary differential equations to describe continuous changes between probability distributions, providing a way to reduce iterative sampling and model segmentation evolution explicitly. Conditioned on an input image $\mathbf{I}$, an initial state sampled from a simple prior distribution $p_0$ is transported through a time-dependent continuous transformation toward the distribution $p_1$ of target segmentation representations. Flow Matching (FM) constructs deterministic transport from $p_0$ to $p_1$ by learning the time-dependent velocity field that drives this distribution change. Let $\psi$ denote the corresponding flow transformation, which satisfies
\begin{equation}
    p_1 = [\psi]_\# p_0,
    \label{eq:pushforward}
\end{equation}
and defines a continuous mapping between the two distributions. From a learning perspective, $\psi$ maps a sample $\mathbf{z}_0\sim p_0$ to $\psi(\mathbf{z}_0)\sim p_1$. In the conditional FM framework, this transformation is parameterized as a time-dependent flow map
\begin{equation}
    \psi_t:[0,1]\times\mathbb{R}^d \rightarrow \mathbb{R}^d,
    \qquad t\in[0,1].
    \label{eq:flow-map-def}
\end{equation}
Let $\mathbf{z}_t\in\mathbb{R}^d$ denote the intermediate state at time $t$. Its evolution is governed by a time-dependent vector field:
\begin{equation}
    \frac{d\mathbf{z}_t}{dt} = \mathbf{v}_\theta(\mathbf{z}_t,t\mid c),
    \label{eq:fm-ode}
\end{equation}
where $\mathbf{v}_\theta$ is the velocity field learned by the parameterized network and $c$ denotes the conditional information. The flow satisfies $\psi_t(\mathbf{z}_0)=\mathbf{z}_t$, and its terminal state is obtained by integration:
\begin{equation}
    \psi(\mathbf{z}_0)=\psi_1(\mathbf{z}_0)
    =\mathbf{z}_0+\int_0^1\mathbf{v}_\theta(\mathbf{z}_t,t\mid c)\,dt.
    \label{eq:flow-integral}
\end{equation}

During training, the complete dynamic trajectory does not need to be solved explicitly. A conditional probability path is constructed between a source endpoint $\mathbf{z}_0$ and a target endpoint $\mathbf{z}_1$. Let $p_t(\mathbf{z}\mid\mathbf{z}_0,\mathbf{z}_1)$ denote this conditional path and $q(\mathbf{z}_0,\mathbf{z}_1)$ the endpoint coupling distribution. The corresponding marginal distribution is
\begin{equation}
    p_t(\mathbf{z})
    =\int p_t(\mathbf{z}\mid\mathbf{z}_0,\mathbf{z}_1)
    q(\mathbf{z}_0,\mathbf{z}_1)\,d\mathbf{z}_0\,d\mathbf{z}_1.
    \label{eq:marginal-path}
\end{equation}
A linear reference path between $\mathbf{z}_0\sim p_0$ and $\mathbf{z}_1\sim p_1$ is commonly used:
\begin{equation}
    \mathbf{z}_t=(1-t)\mathbf{z}_0+t\mathbf{z}_1+\sigma\boldsymbol{\epsilon},
    \qquad \boldsymbol{\epsilon}\sim\mathcal{N}(\mathbf{0},\mathbf{I}),
    \label{eq:linear-bridge}
\end{equation}
with the conditional vector field $\mathbf{v}_t(\mathbf{z}_t\mid\mathbf{z}_0,\mathbf{z}_1)=\mathbf{z}_1-\mathbf{z}_0$. The linear term provides deterministic interpolation between the source and target distributions. The Gaussian term controls the randomness of the conditional path without changing the target velocity. The induced density path and vector field satisfy the continuity equation
\begin{equation}
    \frac{d}{dt}p_t(\mathbf{z})
    +\operatorname{div}\bigl(p_t(\mathbf{z})\mathbf{v}_t(\mathbf{z})\bigr)=0,
    \label{eq:continuity}
\end{equation}
which ensures conservation of probability mass during transport.

Distribution transport can then be reformulated as regression on the target velocity field. The marginal FM objective is
\begin{equation}
    \mathcal{L}_{\mathrm{FM}}(\theta)
    :=\mathbb{E}_{t\sim\mathcal{U}_{[0,1]}}
    \mathbb{E}_{\mathbf{z}\sim p_t}
    \left[
        \frac{1}{2}\left\|
        \mathbf{v}_{\theta,t}(\mathbf{z})-\mathbf{v}_t(\mathbf{z})
        \right\|_2^2
    \right],
    \label{eq:fm-loss}
\end{equation}
where $\mathbf{v}_t$ is the marginal vector field associated with $p_t$. Since $p_t$ and $\mathbf{v}_t$ are difficult to sample and construct directly, conditional FM regresses the conditional vector field along the conditional path:
\begin{equation}
    \mathcal{L}_{\mathrm{CFM}}(\theta)
    =\mathbb{E}_{t,\mathbf{z}_0,\mathbf{z}_1,\mathbf{z}_t}
    \left[
        \frac{1}{2}\left\|
        \mathbf{v}_{\theta,t}(\mathbf{z}_t)
        -\mathbf{v}_t(\mathbf{z}_t\mid\mathbf{z}_0,\mathbf{z}_1)
        \right\|_2^2
    \right].
    \label{eq:cfm-loss}
\end{equation}
The objectives in Eqs.~\eqref{eq:fm-loss} and~\eqref{eq:cfm-loss} have the same optimum. During training, $\mathbf{z}_t$ is sampled directly from Eq.~\eqref{eq:linear-bridge}, and the target velocity is regressed without simulating Eq.~\eqref{eq:fm-ode}. During inference, deterministic numerical integration of the learned ODE transports samples from $p_0$ to $p_1$.

Recent studies have adapted continuous flow models to dense prediction tasks. SemFlow~\cite{wang2024semflow} combines semantic mask prediction and image synthesis within a rectified-flow framework. FlowSDF~\cite{bogensperger2025flowsdf} uses signed distance transforms as continuous targets for conditional flow matching. This design enables the target boundary to be represented as an implicit interface. Diff2Flow~\cite{schusterbauer2025diff2flow} transfers knowledge from diffusion models to flow-based models and reduces the number of sampling steps. These studies focus mainly on transport design, target representation, or sampling efficiency. The spatial differential properties of the recovered implicit field receive less direct attention.

Segmentation objectives regularize predictions at different levels. Cross-entropy provides pixel-wise supervision, while Dice and IoU losses measure regional overlap. Boundary-aware losses further penalize contour discrepancies. SDF-based prediction also permits differential constraints on the recovered field. First-order gradient or Eikonal constraints regulate distance-field properties, whereas curvature- and Laplacian-based terms control higher-order spatial variations. These objectives are usually formulated independently of continuous transport dynamics. Their coupling with velocity regression remains less explored. A unified formulation should therefore connect flow-field learning with the spatial geometry of the recovered SDF.
\section{Proposed Method}
\label{sec:arch-loss}
Building on Sections~\ref{sec:preliminaries} and~\ref{sec:sdf}, this section presents the implementation of the proposed geometrically regularized flow matching model for medical image segmentation. We first formulate segmentation as SDF-based flow matching, then introduce the high-order geometric regularization and image-guided network architecture, and finally describe the inference strategy. The mechanistic properties of the geometry-constrained transport process are analyzed in Section~\ref{sec:loss-mech}.

\subsection{SDF-based Flow Matching}

Consider a training dataset $$\mathcal{D}={(\mathbf{I}^{(n)}, s^{(n)})}_{n=1}^{N}$$, where each training sample consists of an input image $\mathbf{I}^{(n)}$ and its corresponding SDF label $s^{(n)}$. The SDF label is generated from the ground-truth mask $y^{(n)}$ using the distance transform defined in Eq.~\eqref{eq:sdf-def}. Within the flow matching framework of Section~\ref{sec:preliminaries}, the segmentation task is formulated as a continuous transport process in the SDF space. The goal is to learn a conditional vector field $\mathbf{v}_\theta(\mathbf{s}_t,t,\mathbf{I})$ that satisfies the following ordinary differential equation:
\begin{equation}
    \frac{d}{dt}\Psi_t(\mathbf{s}\mid\mathbf{I})
    =
    \mathbf{v}_\theta
    \bigl(
    \Psi_t(\mathbf{s}\mid\mathbf{I}),
    t,
    \mathbf{I}
    \bigr).
    \label{eq:seg-ode}
\end{equation}

During training, the input image $\mathbf{I}$ serves as conditional information that guides the SDF from the initial Gaussian distribution $\mathbf{s}_0\sim p_0=\mathcal{N}(\mathbf{0},\mathbf{I})$ toward the target SDF field $s$. Intermediate states $\mathbf{s}_t$ are sampled from the conditional probability path
$p_t(\mathbf{s}_t\mid\mathbf{s}_0,s,\mathbf{I})$,
thereby constructing supervision signals in continuous time.

Specifically, let
$t\sim\mathcal{U}[0,1]$
and
$\boldsymbol{\epsilon}\sim\mathcal{N}(\mathbf{0},\mathbf{I})$.
The conditional probability path is defined as

\begin{equation}
\mathbf{s}_t
=
\mu_t+\sigma_t\boldsymbol{\epsilon},
\qquad
\mu_t=t\,s,
\qquad
\sigma_t
=
1-(1-\sigma_{\min})t,
\label{eq:affine-path}
\end{equation}

where $\sigma_{\min}\approx0$ denotes the minimum diffusion scale. Expanding Eq.~(\ref{eq:affine-path}) yields

\begin{equation}
\mathbf{s}_t
=
t\,s_1
+
\left(
1-(1-\sigma_{\min})t
\right)
\boldsymbol{\epsilon}.
\label{eq:path-expand}
\end{equation}

Rearranging Eq.~(\ref{eq:path-expand}) gives the Gaussian noise variable as

\begin{equation}
\boldsymbol{\epsilon}
=
\frac{\mathbf{s}_t-t\,s_1}
{1-(1-\sigma_{\min})t}.
\label{eq:epsilon}
\end{equation}

Taking the derivative of Eq.~(\ref{eq:path-expand}) with respect to time gives

\begin{equation}
\frac{\mathrm d\mathbf{s}_t}{\mathrm dt}
=
s_1
-
(1-\sigma_{\min})
\boldsymbol{\epsilon}.
\label{eq:path-derivative}
\end{equation}

Substituting Eq.~(\ref{eq:epsilon}) into Eq.~(\ref{eq:path-derivative}) yields the conditional target velocity

\begin{align}
\mathbf{u}_t
(\mathbf{s}_t\mid s_1,I)
&=
s_1
-
(1-\sigma_{\min})
\frac{\mathbf{s}_t-t\,s_1}
{1-(1-\sigma_{\min})t}
\nonumber\\
&=
\frac{
s_1-(1-\sigma_{\min})\mathbf{s}_t
}
{1-(1-\sigma_{\min})t}.
\label{eq:target-velocity}
\end{align}

Accordingly, the conditional flow matching loss is defined as
\begin{equation}
    \mathcal{L}_{\mathrm{CFM}}
    = \mathbb{E}_{t,\mathbf{I},s,\mathbf{s}_t}
    \left[
        \bigl\|
            \mathbf{v}_\theta(\mathbf{s}_t,t,\mathbf{I})
            - \mathbf{u}_t
        \bigr\|^2
    \right].
    \label{eq:lcfm}
\end{equation}

\subsection{The proposed GR-FM model}
\label{sec:loss}

Building on the FM formulation and SDF representation introduced in Sections~\ref{sec:preliminaries} and~\ref{sec:sdf}, respectively, let the input image be $\mathbf{I}$ and its corresponding SDF label be $s$. As shown in Fig.~\ref{fig:architecture}, a U-Net is used to parameterize the conditional vector field $\mathbf{v}_\theta(\mathbf{s}_t,t,\mathbf{I})$, which progressively transports noisy intermediate states toward the target SDF along a continuous-time trajectory. However, the vector-field regression loss in Eq.~\eqref{eq:lcfm} mainly enforces consistency between the predicted and target velocities and does not directly ensure that the recovered result has desirable spatial geometry. Even with a small velocity regression error, local prediction errors may accumulate through spatial differentiation and discrete integration, leading to local oscillations, rough boundaries, spurious small structures, and uneven gradient distributions.

The geometric information encoded by an SDF is determined not only by its value at each pixel, but also by its spatial gradients, high-order derivatives, and the local shape of its zero level set. Velocity matching in the state space alone cannot adequately control these spatial properties. Applying geometric constraints directly to the noisy state $\mathbf{s}_t$ is also unsuitable because it still contains time-dependent noise, and its spatial variations do not accurately represent the geometry of the target SDF. To apply geometric regularization to a representation closer to the target SDF, an implicit recovery step is first performed on the noisy intermediate state at the current time:

\begin{equation}
\hat{\mathbf{s}}_t
=
(1-\sigma_{\min})\mathbf{s}_t
+
\sigma_t
\mathbf{v}_\theta
\left(
\mathbf{s}_t,t,\mathbf I
\right),
\label{eq:tweedie}
\end{equation}

where $\hat{s}$ denotes the estimate of the target SDF obtained from the current state and the predicted vector field. This recovery step establishes a direct connection between velocity regression and spatial geometric constraints. The training objective can therefore evaluate not only whether the model learns the correct transport direction, but also whether the transported result satisfies the structural properties of the target SDF.

Based on the recovered $\hat{s}$, we construct a new geometrically regularized training objective. The high-order geometric regularization term constrains high-order spatial variations in the SDF, suppressing high-frequency oscillations and improving the local smoothness and structural continuity of the implicit representation. The distance-field constraint regulates the spatial gradient distribution so that the recovered result retains valid SDF properties. These two constraints are jointly optimized with the original vector-field regression loss, regulating SDF recovery in terms of transport direction, spatial regularity, and distance-field properties. This formulation reduces error accumulation during discrete integration and improves boundary evolution, fine-structure preservation, and inference stability with a small number of integration steps.

First, the following biharmonic-type high-order constraint is introduced:
\begin{equation}
    \mathcal{L}_{\mathrm{bal}}
    = \mathbb{E}\bigl[(\Delta \hat{s})^2\bigr],
    \label{eq:lbal}
\end{equation}
where the second-order Laplacian term characterizes the intensity of high-order spatial variation of the implicit representation, thereby suppressing local high-frequency oscillations and irregular fluctuations in the predicted shape and enhancing spatial smoothness and boundary coherence of the SDF representation.

Furthermore, an Eikonal geometric constraint is introduced:
\begin{equation}
    \mathcal{L}_{\mathrm{eik}}
    = \mathbb{E}
    \left[
        \bigl(|\nabla \hat{s}|^2 - 1\bigr)^2
    \right],
\end{equation}
where
\begin{equation}
    |\nabla \hat{s}|^2
    =
    \left(\frac{\partial \hat{s}}{\partial x}\right)^2
    +
    \left(\frac{\partial \hat{s}}{\partial y}\right)^2,
    \label{eq:leik}
\end{equation}
enforces $|\nabla \hat{s}|\approx 1$ so that the recovered implicit representation approximately satisfies the basic properties of a signed distance function. This constraint encourages $\hat{s}$ to maintain a stable distance-field gradient structure near boundaries, thereby improving geometric interpretability and boundary stability near the zero level set, and mitigating excessively flat or steep implicit representations during continuous evolution.

The overall geometric regularization term is therefore
$\mathcal{L}_{\mathrm{geo}}
=
\mathcal{L}_{\mathrm{bal}}
+
\mathcal{L}_{\mathrm{eik}}$.
The joint optimization objective of the network is
\begin{equation}
    \mathcal{L}
    = \mathcal{L}_{\mathrm{CFM}}
    + \lambda_{\mathrm{bal}}\mathcal{L}_{\mathrm{bal}}
    + \lambda_{\mathrm{eik}}\mathcal{L}_{\mathrm{eik}},
    \label{eq:total-loss}
\end{equation}
where $\lambda_{\mathrm{bal}}>0$ and $\lambda_{\mathrm{eik}}>0$ are the weights for the biharmonic and Eikonal constraints, respectively.

\subsection{Image-guided Vector Field Learning}
\label{sec:architecture}

The overall network is built upon a time-conditioned U-Net architecture, as illustrated in Fig.~\ref{fig:architecture}. The noisy SDF state $\mathbf{s}_t$ and the conditional image $\mathbf{I}$ are first fed into a dual-branch encoder to extract multi-scale features. To fully exploit image structure for constraining flow evolution, features from the two branches are fused element-wise at corresponding scales, yielding a joint representation that combines geometric and semantic information. The flow matching time variable $t$ is mapped to a time embedding via an MLP and injected into each layer of the encoder and decoder, enabling the network to learn a conditionally time-varying velocity field.

Subsequently, the fused features pass through downsampling, upsampling, and skip connections in the U-Net backbone to regress a vector field $\mathbf{v}_\theta(\mathbf{s}_t,t,\mathbf{I})$ at the same resolution as $\mathbf{s}_t$, which describes the evolution direction and rate of change of the SDF field in continuous time. To further recover the target implicit representation, a one-step estimate based on the predicted vector field and the current state $\mathbf{s}_t$ is performed during training to obtain the reconstructed SDF $\hat{s}$ (see Eq.~\eqref{eq:tweedie}). The entire training process is carried out in continuous SDF space.

Together, the image-conditioned architecture and the preceding geometric objective define the vector field used by the inference procedure below.

\begin{figure}[ht]
    \centering
    \includegraphics[width=\textwidth]{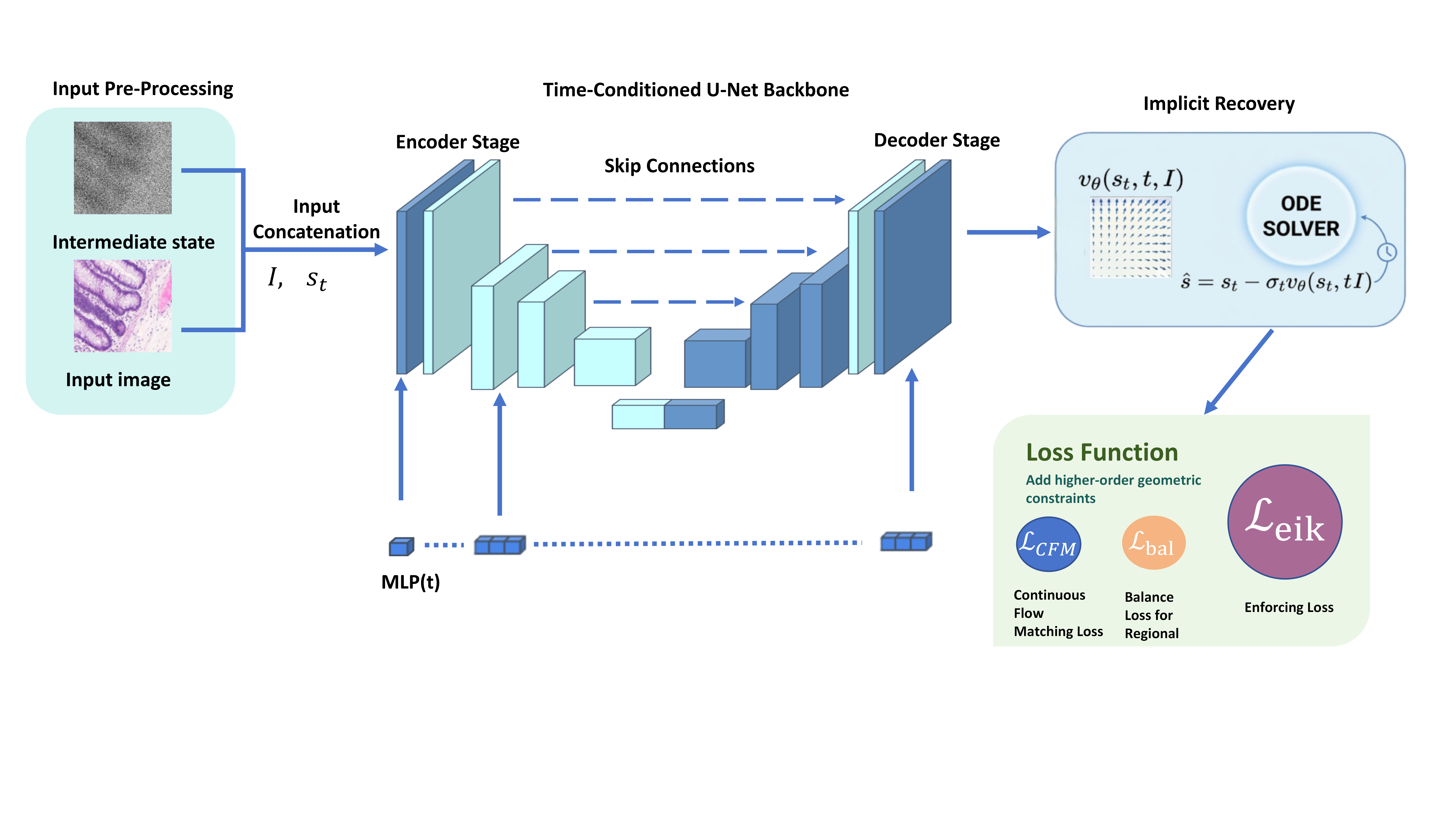}
    \caption{Overall pipeline of the image-guided flow matching segmentation network.}
    \label{fig:architecture}
\end{figure}

\subsection{Inference Strategy}
\label{sec:inference}

At inference, the conditioning image $\mathbf{I}$ and the trained velocity field $\mathbf{v}_{\theta}$ are used to guide the continuous transport process from the prior distribution toward the target SDF distribution. The initial state is sampled from a standard Gaussian prior:
\begin{equation}
\mathbf{s}_0
\sim
p_0
=
\mathcal{N}
\left(
\mathbf{0},
\mathbf{I}_d
\right),
\label{eq:inference-initial-state}
\end{equation}
where $\mathbf{s}_0$ has the same spatial dimensions as the target SDF and $\mathbf{I}_d$ denotes the $d$-dimensional identity matrix. The initial state is retained as a fixed reference throughout the ODE integration, and no additional random noise is introduced during inference.

The ODE is solved using the explicit Euler method. Let $\mathrm{NFE}$ denote the number of vector-field evaluations. The time interval $[0,1]$ is uniformly discretized with step size
\begin{equation}
\Delta t
=
\frac{1}{\mathrm{NFE}}.
\label{eq:inference-step-size}
\end{equation}
The time associated with the $k$-th vector-field evaluation is
\begin{equation}
t_k
=
k\Delta t,
\qquad
k=0,1,\ldots,\mathrm{NFE}-1.
\label{eq:inference-time}
\end{equation}

At time $t_k$, the input state of the velocity field is constructed from the fixed initial state $\mathbf{s}_0$ and the current ODE state $\mathbf{s}_k$:
\begin{equation}
\widetilde{\mathbf{s}}_k
=
\left[
1-
\left(
1-\sigma_{\min}
\right)t_k
\right]
\mathbf{s}_0
+
t_k\mathbf{s}_k,
\label{eq:inference-path-state}
\end{equation}
where $\sigma_{\min}$ denotes the minimum noise scale of the transport path. The path state combines information from the initial prior and the current recovered state, with their contributions changing continuously over time.

The velocity field predicts the SDF evolution direction conditioned on the path state, the current time, and the input image:
\begin{equation}
\mathbf{u}_k
=
\mathbf{v}_{\theta}
\left(
\widetilde{\mathbf{s}}_k,
t_k
\mid
\mathbf{I}
\right).
\label{eq:inference-velocity}
\end{equation}
The current SDF state is updated using the explicit Euler scheme:
\begin{equation}
\mathbf{s}_{k+1}
=
\mathbf{s}_k
+
\Delta t\,\mathbf{u}_k,
\qquad
k=0,1,\ldots,\mathrm{NFE}-1.
\label{eq:inference-euler}
\end{equation}

After $\mathrm{NFE}$ vector-field evaluations, the terminal state $\mathbf{s}_{\mathrm{NFE}}$ is treated as the recovered SDF. Let $\tau$ denote the inference threshold. The final binary segmentation mask is defined as
\begin{equation}
y_{\mathrm{pred}}(\mathbf{x})
=
\mathbf{1}_{\left\{
\mathbf{s}_{\mathrm{NFE}}(\mathbf{x})
\leq
\tau
\right\}}.
\label{eq:inference-binarization}
\end{equation}

The inference process transports a Gaussian prior state toward the target SDF through the image-conditioned velocity field and recovers the final segmentation mask from the terminal SDF. The progressive evolution of the SDF from the initial noise to the target geometric representation is illustrated in Figure~\ref{fig:recovery}.

\begin{algorithm}[ht]
    \caption{Image-Guided SDF Sampling}
    \label{alg:image-guided-sampling}
    \begin{algorithmic}[1]
        \Require Trained velocity field $\mathbf{v}_{\theta}$; conditioning image $\mathbf{I}$; number of function evaluations $\mathrm{NFE}$; minimum noise scale $\sigma_{\min}$; inference threshold $\delta$
        
        \State $\mathbf{s}_0
        \sim
        \mathcal{N}
        \left(
        \mathbf{0},
        \mathbf{I}_d
        \right)$

        \State $\Delta t
        =
        1/\mathrm{NFE}$

        \For{$k=0,1,\ldots,\mathrm{NFE}-1$}
            \State $t_k = k\Delta t$

            \State $\widetilde{\mathbf{s}}_k
            \gets
            \left[
            1-
            \left(
            1-\sigma_{\min}
            \right)t_k
            \right]
            \mathbf{s}_0
            +
            t_k\mathbf{s}_k$

            \State $\mathbf{s}_{k+1}
            \gets
            \mathbf{s}_k
            +
            \Delta t\,
            \mathbf{v}_{\theta}
            \left(
            \widetilde{\mathbf{s}}_k,
            t_k
            \mid
            \mathbf{I}
            \right)$
        \EndFor

        \State $y_{\mathrm{pred}}(\mathbf{x})
        \gets
        \mathbf{1}_{\left\{
        \mathbf{s}_{\mathrm{NFE}}(\mathbf{x})
        \leq
        \tau
        \right\}}$

        \State \Return
        $\mathbf{s}_{\mathrm{NFE}},
        y_{\mathrm{pred}}$
    \end{algorithmic}
\end{algorithm}

\begin{figure}[ht]
    \centering
    \includegraphics[width=\textwidth]{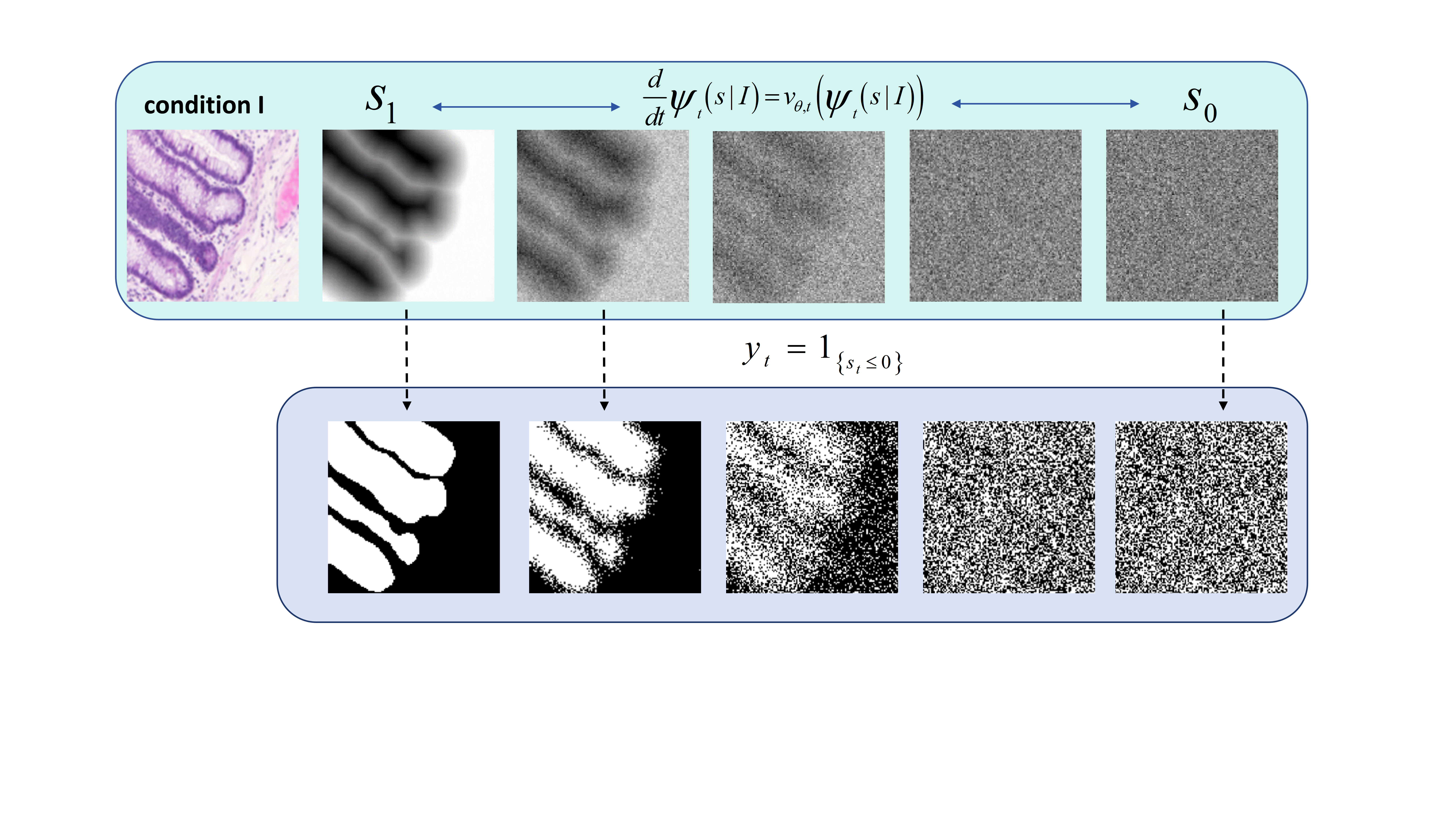}
    \caption{The figure depicts the forward corruption and reverse generation process of the SDF mask $m$ at varying timesteps $t \in [0,1]$. The bottom row shows the derived binary masks $m_t$, illustrating the underlying deformation mechanism governed by the SDF.}
    \label{fig:recovery}
\end{figure}
\section{Theoretical Analysis}
\label{sec:loss-mech}

This paragraph clarifies the role of geometric regularization in Eq.~\eqref{eq:total-loss} from two perspectives. First, because the transported state is an SDF, the flow evolution in this work can be interpreted as level-set transport; within this framework, we analyze the transport properties of the zero level set along continuous flow evolution and the effect of the residual $r$ on interface stability. Second, after adding $\mathcal{L}_{\mathrm{geo}}$, conditional flow matching loss and marginal flow matching loss remain equivalent in the optimization sense.
\subsection{Zero-Level Set Transport Analysis}

The SDF evolution is analyzed within a level-set framework. Let
$\Omega\subset\mathbb{R}^{d}$, $t\in[0,T]$, and
$s:[0,T]\times\Omega\to\mathbb{R}$. The continuous flow determined by
Eq.~\eqref{eq:seg-ode} is written as
\begin{equation}
    \partial_t s(t,\mathbf{x})
    +
    \mathbf{v}(t,\mathbf{x})\cdot\nabla s(t,\mathbf{x})
    =
    r(t,\mathbf{x}),
    \label{eq:sdf-transport}
\end{equation}
where
$\mathbf{v}:[0,T]\times\Omega\to\mathbb{R}^{d}$
is the level-set transport velocity field, and
$r:[0,T]\times\Omega\to\mathbb{R}$
is a residual source term characterizing the deviation of the learned
SDF evolution from exact level-set advection.

Define the zero level set by
\begin{equation}
    \Gamma_t
    :=
    \left\{
    \mathbf{x}\in\Omega:
    s(t,\mathbf{x})=0
    \right\}.
    \label{eq:zero-level-set}
\end{equation}
According to the continuous Flow Matching dynamics, the pointwise
temporal variation of the SDF satisfies
\begin{equation}
    \partial_t s(t,\mathbf{x})
    =
    \mathbf{v}_{\theta}
    \bigl(
    s_t,t,\mathbf{I}
    \bigr)(\mathbf{x}),
    \qquad
    s_t=s(t,\cdot),
    \label{eq:pointwise-sdf-evolution}
\end{equation}
where
$\mathbf{v}_{\theta}(s_t,t,\mathbf{I})(\mathbf{x})$
denotes the predicted rate of change of the SDF at the spatial location
$\mathbf{x}$.

Consider the characteristic curve
\begin{equation}
    \frac{d\mathbf{X}(t)}{dt}
    =
    \mathbf{v}
    \bigl(
    t,\mathbf{X}(t)
    \bigr),
    \qquad
    \mathbf{X}(0)=\mathbf{x}_0.
    \label{eq:characteristic}
\end{equation}
Whenever
$\nabla s(t,\mathbf{X}(t))\neq\mathbf{0}$,
define the unit normal vector by
\begin{equation}
    \mathbf{n}(t)
    :=
    \frac{
    \nabla s
    \bigl(
    t,\mathbf{X}(t)
    \bigr)
    }{
    \left\|
    \nabla s
    \bigl(
    t,\mathbf{X}(t)
    \bigr)
    \right\|
    },
    \label{eq:unit-normal}
\end{equation}
and the normal component of the velocity field by
\begin{equation}
    v_{\perp}(t)
    :=
    \mathbf{v}
    \bigl(
    t,\mathbf{X}(t)
    \bigr)
    \cdot
    \mathbf{n}(t).
    \label{eq:normal-velocity}
\end{equation}
\begin{theorem}[Zero-Level Set Transport]
\label{thm:levelset-transport}

Suppose that $s$ satisfies Eq.~\eqref{eq:sdf-transport}, and let $\mathbf{X}(t)$ be the characteristic curve defined by Eq.~\eqref{eq:characteristic}. Then the evolution of the SDF along the characteristic satisfies
\begin{equation}
\frac{d}{dt}s\bigl(t,\mathbf{X}(t)\bigr)=r\bigl(t,\mathbf{X}(t)\bigr).
\label{eq:characteristic-sdf}
\end{equation}

Moreover, the normal velocity of the zero level set is given by
\begin{equation}
v_{\perp}(t)=-\frac{\mathbf{v}_{\theta}\bigl(s_t,t,\mathbf{I}\bigr)\bigl(\mathbf{X}(t)\bigr)}{\left\|\nabla s\bigl(t,\mathbf{X}(t)\bigr)\right\|}+\frac{r\bigl(t,\mathbf{X}(t)\bigr)}{\left\|\nabla s\bigl(t,\mathbf{X}(t)\bigr)\right\|},
\label{eq:normal-velocity-with-residual}
\end{equation}
provided that $\left\|\nabla s\bigl(t,\mathbf{X}(t)\bigr)\right\|\neq0$ in a neighborhood of the zero level set.

If the residual term vanishes identically, that is, $r\equiv0$, then the SDF value remains constant along each characteristic:
\begin{equation}
s\bigl(t,\mathbf{X}(t)\bigr)=s(0,\mathbf{x}_0), \qquad \forall\,t\in[0,T].
\label{eq:sdf-conservation}
\end{equation}

In particular, if the initial point lies on the initial zero level set, namely,
\begin{equation}
\mathbf{x}_0\in\Gamma_0,
\end{equation}
then
\begin{equation}
\mathbf{X}(t)\in\Gamma_t, \qquad \forall\,t\in[0,T].
\label{eq:exact-interface-transport}
\end{equation}

Therefore, in the absence of the residual term, the zero level set is transported exactly by the continuous flow.

If $r\not\equiv0$, then
\begin{equation}
s\bigl(t,\mathbf{X}(t)\bigr)=s(0,\mathbf{x}_0)+\int_0^t r\bigl(\tau,\mathbf{X}(\tau)\bigr)\,d\tau.
\label{eq:residual-integral}
\end{equation}

Furthermore, if $\mathbf{x}_0\in\Gamma_0$, then
\begin{equation}
\left|s\bigl(t,\mathbf{X}(t)\bigr)\right|\le\int_0^t\left|r\bigl(\tau,\mathbf{X}(\tau)\bigr)\right|\,d\tau.
\label{eq:residual-bound}
\end{equation}

Hence, the deviation from the exact transport of the zero level set is bounded by the accumulated residual along the characteristic.

\end{theorem}

\begin{proof}

Along the characteristic $\mathbf{X}(t)$, the chain rule gives
\begin{equation}
\frac{d}{dt}s\bigl(t,\mathbf{X}(t)\bigr)=\partial_t s\bigl(t,\mathbf{X}(t)\bigr)+\nabla s\bigl(t,\mathbf{X}(t)\bigr)\cdot\frac{d\mathbf{X}(t)}{dt}.
\end{equation}
Together with $\frac{d\mathbf{X}(t)}{dt}=\mathbf{v}\bigl(t,\mathbf{X}(t)\bigr)$ from Eq.~\eqref{eq:characteristic}, we obtain
\begin{equation}
\frac{d}{dt}s\bigl(t,\mathbf{X}(t)\bigr)=\partial_t s\bigl(t,\mathbf{X}(t)\bigr)+\mathbf{v}\bigl(t,\mathbf{X}(t)\bigr)\cdot\nabla s\bigl(t,\mathbf{X}(t)\bigr)=r\bigl(t,\mathbf{X}(t)\bigr),
\end{equation}
where the last equality follows from Eq.~\eqref{eq:sdf-transport}, proving Eq.~\eqref{eq:characteristic-sdf}.

Next, by the definitions of the unit normal and the normal velocity,
\begin{equation}
\mathbf{n}(t)=\frac{\nabla s\bigl(t,\mathbf{X}(t)\bigr)}{\left\|\nabla s\bigl(t,\mathbf{X}(t)\bigr)\right\|}, \qquad v_{\perp}(t)=\mathbf{v}\bigl(t,\mathbf{X}(t)\bigr)\cdot\mathbf{n}(t),
\end{equation}
we have
\begin{equation}
\mathbf{v}\bigl(t,\mathbf{X}(t)\bigr)\cdot\nabla s\bigl(t,\mathbf{X}(t)\bigr)=v_{\perp}(t)\left\|\nabla s\bigl(t,\mathbf{X}(t)\bigr)\right\|.
\end{equation}
Combining Eq.~\eqref{eq:pointwise-sdf-evolution}, Eq.~\eqref{eq:sdf-transport}, and the identity above gives
\begin{equation}
r\bigl(t,\mathbf{X}(t)\bigr)=\mathbf{v}_{\theta}\bigl(s_t,t,\mathbf{I}\bigr)\bigl(\mathbf{X}(t)\bigr)+v_{\perp}(t)\left\|\nabla s\bigl(t,\mathbf{X}(t)\bigr)\right\|.
\end{equation}
Therefore, when $\left\|\nabla s\bigl(t,\mathbf{X}(t)\bigr)\right\|\neq0$,
\begin{equation}
v_{\perp}(t)=-\frac{\mathbf{v}_{\theta}\bigl(s_t,t,\mathbf{I}\bigr)\bigl(\mathbf{X}(t)\bigr)}{\left\|\nabla s\bigl(t,\mathbf{X}(t)\bigr)\right\|}+\frac{r\bigl(t,\mathbf{X}(t)\bigr)}{\left\|\nabla s\bigl(t,\mathbf{X}(t)\bigr)\right\|},
\end{equation}
which proves Eq.~\eqref{eq:normal-velocity-with-residual}.

(1) If $r\equiv0$, then Eq.~\eqref{eq:characteristic-sdf} reduces to
\begin{equation}
\frac{d}{dt}s\bigl(t,\mathbf{X}(t)\bigr)=0.
\end{equation}
Hence,
\begin{equation}
s\bigl(t,\mathbf{X}(t)\bigr)=s(0,\mathbf{x}_0),
\end{equation}
which is Eq.~\eqref{eq:sdf-conservation}. If $\mathbf{x}_0\in\Gamma_0$, then $s(0,\mathbf{x}_0)=0$, and thus
\begin{equation}
s\bigl(t,\mathbf{X}(t)\bigr)=0.
\end{equation}
Therefore,
\begin{equation}
\mathbf{X}(t)\in\Gamma_t, \qquad \forall\,t\in[0,T],
\end{equation}
which proves Eq.~\eqref{eq:exact-interface-transport}.

(2) If $r\not\equiv0$, then
\begin{equation}
\frac{d}{dt}s\bigl(t,\mathbf{X}(t)\bigr)=r\bigl(t,\mathbf{X}(t)\bigr).
\end{equation}
Integrating over $[0,t]$ gives
\begin{equation}
s\bigl(t,\mathbf{X}(t)\bigr)-s(0,\mathbf{x}_0)=\int_0^t r\bigl(\tau,\mathbf{X}(\tau)\bigr)\,d\tau,
\end{equation}
which is Eq.~\eqref{eq:residual-integral}. If $\mathbf{x}_0\in\Gamma_0$, then $s(0,\mathbf{x}_0)=0$, so
\begin{equation}
s\bigl(t,\mathbf{X}(t)\bigr)=\int_0^t r\bigl(\tau,\mathbf{X}(\tau)\bigr)\,d\tau.
\end{equation}
By the triangle inequality,
\begin{equation}
\left|s\bigl(t,\mathbf{X}(t)\bigr)\right|\le\int_0^t\left|r\bigl(\tau,\mathbf{X}(\tau)\bigr)\right|\,d\tau,
\end{equation}
which proves Eq.~\eqref{eq:residual-bound}.

Finally, if the Eikonal condition approximately holds near the zero level set, namely, $\|\nabla s\|\approx1$, then Eq.~\eqref{eq:normal-velocity-with-residual} becomes
\begin{equation}
v_{\perp}(t)\approx-\mathbf{v}_{\theta}\bigl(s_t,t,\mathbf{I}\bigr)\bigl(\mathbf{X}(t)\bigr)+r\bigl(t,\mathbf{X}(t)\bigr).
\end{equation}
Thus, when the residual is small, the predicted SDF evolution provides an approximation to the normal motion of the zero level set.

\end{proof}

\begin{remark}
A nonzero residual $r$ drives trajectories away from the zero level set and induces interface drift. When only $\mathcal{L}_{\mathrm{CFM}}$ is optimized, $\mathbf{v}_\theta$ is regressed on noisy $\mathbf{s}_t$, and it is difficult to ensure $r\approx 0$; zero level set instability therefore tends to arise under few-step ODE integration. $\mathcal{L}_{\mathrm{bal}}$ and $\mathcal{L}_{\mathrm{eik}}$ act on $\hat{s}$ to constrain the implicit interface from the perspectives of smoothness and the Eikonal structure $|\nabla s|\approx 1$, respectively, which helps suppress the effective residual $r$ and improves the smoothness of evolution trajectories and integration stability at low NFE.
\end{remark}

\subsection{Marginal Regression Equivalence}
Within the flow matching framework of Section~\ref{sec:preliminaries}, given image $\mathbf{I}$ and SDF label $s$, sample $t\sim\mathcal{U}[0,1]$ and obtain $\mathbf{s}_t$ from the conditional path $p_t(\mathbf{s}_t\mid s,\mathbf{I})$ (Eq.~\eqref{eq:affine-path}), with conditional target velocity $\mathbf{u}_t(\mathbf{s}_t\mid s,\mathbf{I})$ (Eq.~\eqref{eq:target-velocity}). The network outputs $\mathbf{v}_\theta(\mathbf{s}_t,t,\mathbf{I})$ and applies geometric regularization on $\hat{s}$ (Eq.~\eqref{eq:tweedie}),
$\mathcal{L}_{\mathrm{geo}}(\mathbf{s}_t,t,\mathbf{I};\mathbf{v}_\theta)=\mathcal{L}_{\mathrm{bal}}+\mathcal{L}_{\mathrm{eik}}$ (Eqs.~\eqref{eq:lbal} and~\eqref{eq:leik}).
Define the marginal vector field
\begin{equation}
    \bar{\mathbf{u}}_t(\mathbf{s}_t,t,\mathbf{I})
    := \mathbb{E}\bigl[\mathbf{u}_t(\mathbf{s}_t\mid s,\mathbf{I})\mid \mathbf{s}_t,t,\mathbf{I}\bigr],
    \label{eq:marginal-velocity}
\end{equation}
and
\begin{equation}
    \mathcal{L}_{\mathrm{cond}}(\theta)
    = \mathbb{E}
    \left[
        \bigl\|
            \mathbf{v}_\theta(\mathbf{s}_t,t,\mathbf{I})
            - \mathbf{u}_t(\mathbf{s}_t\mid s,\mathbf{I})
        \bigr\|^2
    \right]
    + \mathbb{E}\bigl[\mathcal{L}_{\mathrm{geo}}(\mathbf{s}_t,t,\mathbf{I};\mathbf{v}_\theta)\bigr],
    \label{eq:lcond-geo}
\end{equation}
\begin{equation}
    \mathcal{L}_{\mathrm{marg}}(\theta)
    = \mathbb{E}
    \left[
        \bigl\|
            \mathbf{v}_\theta(\mathbf{s}_t,t,\mathbf{I})
            - \bar{\mathbf{u}}_t(\mathbf{s}_t,t,\mathbf{I})
        \bigr\|^2
    \right]
    + \mathbb{E}\bigl[\mathcal{L}_{\mathrm{geo}}(\mathbf{s}_t,t,\mathbf{I};\mathbf{v}_\theta)\bigr].
    \label{eq:lmarg-geo}
\end{equation}

\begin{theorem}[Marginal regression equivalence]
    \label{thm:marginal-equiv}
    For any parameter $\theta$, there exists a constant $C\ge 0$ independent of $\theta$ such that
    \begin{equation}
        \mathcal{L}_{\mathrm{cond}}(\theta)
        = \mathcal{L}_{\mathrm{marg}}(\theta) + C,
        \label{eq:cond-marg-equiv}
    \end{equation}
    where
    \begin{equation}
        C
        = \mathbb{E}
        \left[
            \bigl\|
                \mathbf{u}_t(\mathbf{s}_t\mid s,\mathbf{I})
                - \bar{\mathbf{u}}_t(\mathbf{s}_t,t,\mathbf{I})
            \bigr\|^2
        \right].
        \label{eq:cond-marg-C}
    \end{equation}
\end{theorem}

\begin{proof}
    Since
    $\bar{\mathbf{u}}_t(\mathbf{s}_t,t,\mathbf{I})
    =\mathbb{E}[\mathbf{u}_t(\mathbf{s}_t\mid s,\mathbf{I})\mid\mathbf{s}_t,t,\mathbf{I}]$,
    we have
    \begin{equation}
        \mathbf{v}_\theta(\mathbf{s}_t,t,\mathbf{I})
        -\mathbf{u}_t(\mathbf{s}_t\mid s,\mathbf{I})
        =
        \bigl(\mathbf{v}_\theta(\mathbf{s}_t,t,\mathbf{I})
        -\bar{\mathbf{u}}_t(\mathbf{s}_t,t,\mathbf{I})\bigr)
        +
        \bigl(\bar{\mathbf{u}}_t(\mathbf{s}_t,t,\mathbf{I})
        -\mathbf{u}_t(\mathbf{s}_t\mid s,\mathbf{I})\bigr).
    \end{equation}
    Squaring both sides and taking expectations yields
    \begin{equation}
        \begin{aligned}
            &\mathbb{E}
            \bigl[
                \bigl\|
                    \mathbf{v}_\theta(\mathbf{s}_t,t,\mathbf{I})
                    -\mathbf{u}_t(\mathbf{s}_t\mid s,\mathbf{I})
                \bigr\|^2
            \bigr] \\
            ={}& \mathbb{E}
            \bigl[
                \bigl\|
                    \mathbf{v}_\theta(\mathbf{s}_t,t,\mathbf{I})
                    -\bar{\mathbf{u}}_t(\mathbf{s}_t,t,\mathbf{I})
                \bigr\|^2
            \bigr]
            +
            \mathbb{E}
            \bigl[
                \bigl\|
                    \bar{\mathbf{u}}_t(\mathbf{s}_t,t,\mathbf{I})
                    -\mathbf{u}_t(\mathbf{s}_t\mid s,\mathbf{I})
                \bigr\|^2
            \bigr] \\
            &\quad
            + 2\,\mathbb{E}
            \bigl[
                \bigl\langle
                    \mathbf{v}_\theta(\mathbf{s}_t,t,\mathbf{I})
                    -\bar{\mathbf{u}}_t(\mathbf{s}_t,t,\mathbf{I}),\,
                    \bar{\mathbf{u}}_t(\mathbf{s}_t,t,\mathbf{I})
                    -\mathbf{u}_t(\mathbf{s}_t\mid s,\mathbf{I})
                \bigr\rangle
            \bigr].
        \end{aligned}
        \label{eq:error-decomp}
    \end{equation}
    Because $\mathbf{v}_\theta(\mathbf{s}_t,t,\mathbf{I})
    -\bar{\mathbf{u}}_t(\mathbf{s}_t,t,\mathbf{I})$
    is measurable with respect to $(\mathbf{s}_t,t,\mathbf{I})$, and
    $\mathbb{E}[\bar{\mathbf{u}}_t(\mathbf{s}_t,t,\mathbf{I})
    -\mathbf{u}_t(\mathbf{s}_t\mid s,\mathbf{I})
    \mid\mathbf{s}_t,t,\mathbf{I}]=\mathbf{0}$,
    we have
    \begin{equation}
        \mathbb{E}
        \bigl[
            \bigl\langle
                \mathbf{v}_\theta(\mathbf{s}_t,t,\mathbf{I})
                -\bar{\mathbf{u}}_t(\mathbf{s}_t,t,\mathbf{I}),\,
                \bar{\mathbf{u}}_t(\mathbf{s}_t,t,\mathbf{I})
                -\mathbf{u}_t(\mathbf{s}_t\mid s,\mathbf{I})
            \bigr\rangle
        \bigr]=0.
    \end{equation}
    Moreover, $\mathcal{L}_{\mathrm{geo}}$ does not depend on $s$, and hence
    \begin{equation}
        \mathbb{E}
        \bigl[
            \bigl\|
                \mathbf{v}_\theta(\mathbf{s}_t,t,\mathbf{I})
                -\mathbf{u}_t(\mathbf{s}_t\mid s,\mathbf{I})
            \bigr\|^2
        \bigr]
        =
        \mathbb{E}
        \bigl[
            \bigl\|
                \mathbf{v}_\theta(\mathbf{s}_t,t,\mathbf{I})
                -\bar{\mathbf{u}}_t(\mathbf{s}_t,t,\mathbf{I})
            \bigr\|^2
        \bigr]+C.
    \end{equation}
    Therefore $\mathcal{L}_{\mathrm{cond}}(\theta)=\mathcal{L}_{\mathrm{marg}}(\theta)+C$, with $C$ independent of $\theta$.
\end{proof}

\begin{remark}
$\mathcal{L}_{\mathrm{geo}}$ is independent of $s$.
$\mathcal{L}_{\mathrm{cond}}(\theta)=\mathcal{L}_{\mathrm{marg}}(\theta)+C$, and $C$ is independent of $\theta$; the additional high-order geometric constraints do not break the basic equivalence in flow matching between regressing the conditional velocity field and regressing the marginal velocity field.
\end{remark}
\section{Experiments}

\subsection{Datasets and evolution metrics}\label{subsec:datasets}

We evaluate our method on three public medical benchmarks--MoNuSeg, GlaS, and DRIVE, in order to assess generalization under small-sample training and diverse anatomical structures.

\textbf{(1) MoNuSeg dataset}~\cite{kumar2017dataset}.
The dataset contains 24 training images and 4 test images, and is mainly used to evaluate segmentation performance in a low-data regime.
Each image has a resolution of $1000\times1000$ and more than 21{,}000 annotated nuclei in H\&E-stained histopathology microscopy images.
Because images from different organ sites exhibit substantial intensity variation, we apply structure-preserving color normalization~\cite{vahadane2016structure} during preprocessing;
all images are then resized to $500\times500$ following~\cite{li2021semantic}.
During training, we use random crops of $256\times256$ and do \textbf{not} apply additional data augmentation.
We adopt the official train/test split (24 training vs.\ 4 test images) without further random re-partitioning.

\textbf{(2) Gland Segmentation (GlaS) dataset}~\cite{sirinukunwattana2017glas}.
The dataset provides 85 training and 80 test images under the official split.
All images are H\&E-stained microscopy slides of colorectal cancer tissue with gland annotations.
In preprocessing, all training and test images are resized to $128\times128$ and normalized using the same structure-preserving scheme as in~\cite{vahadane2016structure}.
In our experimental protocol, the official training and test sets are merged and re-split at a train:test ratio of $1{:}4$, yielding 33 training and 132 test images.
Training uses random crops of $128\times128$.

\textbf{(3) DRIVE dataset}~\cite{staal2004drive}.
DRIVE is a standard benchmark for retinal vessel segmentation in color fundus images, with 40 training and 20 test images at $565\times584$ resolution, and is used to evaluate boundary delineation of thin tubular structures.
Images and labels are resized to $512\times512$ in preprocessing.
Training employs $128\times128$ random crops with the same background filtering strategy as on GlaS .
At test time, inference is performed on full images; when the resolution exceeds the training crop size, we use a sliding-window predictor and average predictions in overlapping regions.

\paragraph{Evaluation metrics}

Let $s_{\mathrm{pred}}$ and $s_{\mathrm{gt}}$ denote the predicted and ground-truth SDFs, respectively, and let $\tau$ denote the inference threshold. The thresholded predicted foreground region $P$ and ground-truth foreground region $G$ are defined as
\begin{equation}
P
=
\left\{
\mathbf{x}\in\Omega
\mid
s_{\mathrm{pred}}(\mathbf{x})\leq \tau
\right\},
\qquad
G
=
\left\{
\mathbf{x}\in\Omega
\mid
s_{\mathrm{gt}}(\mathbf{x})\leq 0
\right\}.
\end{equation}

The Dice coefficient is defined as
\begin{equation}
\mathrm{Dice}
=
\frac{2|P\cap G|}{|P|+|G|},
\end{equation}
where $|P\cap G|$ denotes the number of pixels in the intersection of the predicted and ground-truth regions. A higher Dice score indicates a larger overlap between the two regions.

The Intersection over Union (IoU) is defined as
\begin{equation}
\mathrm{IoU}
=
\frac{|P\cap G|}{|P\cup G|},
\end{equation}
where $|P\cup G|$ denotes the number of pixels in the union of the predicted and ground-truth regions. A higher IoU value indicates better pixel-level segmentation performance.

Boundary accuracy is evaluated using the 95th-percentile Hausdorff distance (HD$_{95}$). Let $\partial P$ and $\partial G$ denote the boundaries of the predicted and ground-truth regions, respectively. For a predicted boundary point $p\in\partial P$ and a ground-truth boundary point $g\in\partial G$, the point-wise Euclidean distance is defined as
\begin{equation}
d(p,g)
=
\left\|p-g\right\|_2.
\end{equation}

Based on this point-wise Euclidean distance, the conventional Hausdorff distance between the two boundary sets is defined as
\begin{equation}
H(\partial P,\partial G)
=
\max
\left\{
\sup_{p\in\partial P}
\inf_{g\in\partial G}
d(p,g),
\;
\sup_{g\in\partial G}
\inf_{p\in\partial P}
d(g,p)
\right\}.
\end{equation}

HD$_{95}$ replaces the maximum directed distances in the conventional Hausdorff distance with their 95th percentiles. The two directed boundary-distance sets are defined as
\begin{equation}
D_{P\rightarrow G}
=
\left\{
\inf_{g\in\partial G}d(p,g)
\mid
p\in\partial P
\right\},
\qquad
D_{G\rightarrow P}
=
\left\{
\inf_{p\in\partial P}d(g,p)
\mid
g\in\partial G
\right\}.
\end{equation}
HD$_{95}$ is then calculated as
\begin{equation}
\mathrm{HD}_{95}(P,G)
=
\max
\left\{
Q_{0.95}\left(D_{P\rightarrow G}\right),
\;
Q_{0.95}\left(D_{G\rightarrow P}\right)
\right\},
\end{equation}
where $Q_{0.95}$ denotes the 95th percentile of a directed boundary-distance set. A lower HD$_{95}$ value indicates closer agreement between the predicted and ground-truth boundaries.

In the experimental tables, $\uparrow$ and $\downarrow$ indicate that higher and lower values are preferred, respectively.

\subsection{Implementation Details}
\paragraph{Training details}
The deep network that approximates the vector field $v_{\theta}(\cdot,t,x)$ follows a U-Net-style design.
For conditioning, we adopt the image-guided mechanism of SegDiff~\cite{amit2021segdiff}: features extracted from the noisy mask $s_t$ are fused with feature maps of the conditional image $x$ via element-wise addition.
The network takes the input image $x$, random time $t$, and intermediate state $s_t$ as inputs and learns to approximate the target velocity field $v_\theta(s_t,t,x)$.
Supervision is provided by truncated SDF representations $\tilde{m}$ converted from segmentation annotations.

We train with the Adam optimizer at an initial learning rate of $1\times10^{-4}$ and a batch size of 4.
By default, each dataset is trained for $10{,}000$ epochs.
Exponential moving average (EMA) with decay $0.9999$ is applied throughout training for stable parameter updates.
All experiments are conducted on a single NVIDIA Titan RTX 6000 GPU.
Let $S_{\mathrm{ep}}$ denote the number of iterations per epoch (determined by dataset size and batch size); the total number of optimization steps is approximately $10{,}000\times S_{\mathrm{ep}}$.
Specifically, on MoNuSeg each epoch takes about $3.8\,\mathrm{s}$ with $S_{\mathrm{ep}}=6$ ($\sim 6.0\times10^4$ steps in total);
on GlaS each epoch takes about $7.1\,\mathrm{s}$ with $S_{\mathrm{ep}}\approx 9$;
on DRIVE each epoch takes about $11.7\,\mathrm{s}$ with $S_{\mathrm{ep}}=10$ ($\sim 1.0\times10^5$ steps in total).
Unless otherwise stated, all other hyperparameters are shared across datasets.

\paragraph{Inference details}
At test time, we sample an initial SDF state $s_0$ from a standard Gaussian prior and integrate the learned velocity field with an ODE solver to obtain the final SDF prediction $\tilde{m}_1$.
The binary mask is recovered by zero-level thresholding, $m=\mathbf{1}_{\{\tilde{m}_1\leq \tau\}}$.
We use an Euler ODE solver by default; the NFE, i.e., discretization steps, is denoted by $N_{\mathrm{step}}$.
When multiple independent samples are drawn, we average $K$ SDF predictions to approximate the minimum mean squared error estimate before thresholding for metric computation.

\subsection{Ablation Study}\label{subsec:ablation}

To evaluate the contribution of each geometric constraint, we conduct an ablation study based on the FlowSDF baseline. Starting from the SDF-based flow matching model, we progressively add the biharmonic regularization term $\mathcal{L}_{\mathrm{bal}}$ and the Eikonal constraint $\mathcal{L}_{\mathrm{eik}}$. The full model, denoted as Ours, combines both geometric constraints in the joint optimization objective. This comparison is designed to isolate the effects of high-order smoothness and distance-field consistency on segmentation overlap and boundary localization. All results are obtained by repeating the experiments with different random seeds, namely $0$, $42$, $666$, $888$, and $3407$, and are reported in terms of mean and standard deviation, thereby reflecting both average performance and robustness to randomness.

\begin{table}[H]
  \centering
  \caption{Ablation study of geometric regularization terms on GlaS, MoNuSeg, and DRIVE.}
  \label{tab:ablation_geo}
  \footnotesize
  \setlength{\tabcolsep}{8pt}
  \renewcommand{\arraystretch}{1.12}
  \begin{tabular}{|c|l|c|c|c|}
    \hline
    Dataset & Method & Dice$\uparrow$ & IoU/mIoU$\uparrow$ & HD$_{95}$$\downarrow$ \\
    \hline
    \multirow{4}{*}{GlaS}
    & FlowSDF & $0.9346\pm0.0342$ & $0.8791\pm0.0575$ & $9.2783\pm7.5068$ \\
    & FlowSDF+$\mathcal{L}_{\mathrm{bal}}$ & $0.9412\pm0.0080$ & $0.8890\pm0.0142$ & $\underline{7.4505\pm3.1695}$ \\
    & FlowSDF+$\mathcal{L}_{\mathrm{eik}}$ & $\underline{0.9414\pm0.0073}$ & $\underline{0.8894\pm0.0130}$ & $\mathbf{6.4344\pm2.6099}$ \\
    & Ours & $\mathbf{0.9419\pm0.0077}$ & $\mathbf{0.8903\pm0.0138}$ & $7.4821\pm2.6826$ \\
    \hline
    \multirow{4}{*}{MoNuSeg}
    & FlowSDF & $\mathbf{0.8690\pm0.0286}$ & $0.6255\pm0.0379$ & $5.0725\pm1.2434$ \\
    & FlowSDF+$\mathcal{L}_{\mathrm{bal}}$ & $0.8447\pm0.0281$ & $0.6295\pm0.0357$ & $\underline{4.1377\pm0.6378}$ \\
    & FlowSDF+$\mathcal{L}_{\mathrm{eik}}$ & $\underline{0.8689\pm0.0187}$ & $\underline{0.6297\pm0.0246}$ & $4.7690\pm0.5414$ \\
    & Ours & $0.8623\pm0.0113$ & $\mathbf{0.6361\pm0.0148}$ & $\mathbf{4.0983\pm0.3682}$ \\
    \hline
    \multirow{4}{*}{DRIVE}
    & FlowSDF & $0.7781\pm0.0226$ & $0.6374\pm0.0299$ & $13.0054\pm3.4181$ \\
    & FlowSDF+$\mathcal{L}_{\mathrm{bal}}$ & $0.8624\pm0.0095$ & $0.7582\pm0.0147$ & $4.0881\pm1.2311$ \\
    & FlowSDF+$\mathcal{L}_{\mathrm{eik}}$ & $\underline{0.8896\pm0.0128}$ & $\underline{0.8013\pm0.0207}$ & $\underline{2.7188\pm1.1700}$ \\
    & Ours & $\mathbf{0.9253\pm0.0208}$ & $\mathbf{0.8616\pm0.0376}$ & $\mathbf{1.3652\pm0.8463}$ \\
    \hline
  \end{tabular}
\end{table}

To complement the numerical results in Table~\ref{tab:ablation_geo}, we further visualize the variation of the boundary-sensitive HD$_{95}$ metric across the five random seeds.

\begin{figure}[H]
    \centering
    \includegraphics[width=\textwidth]{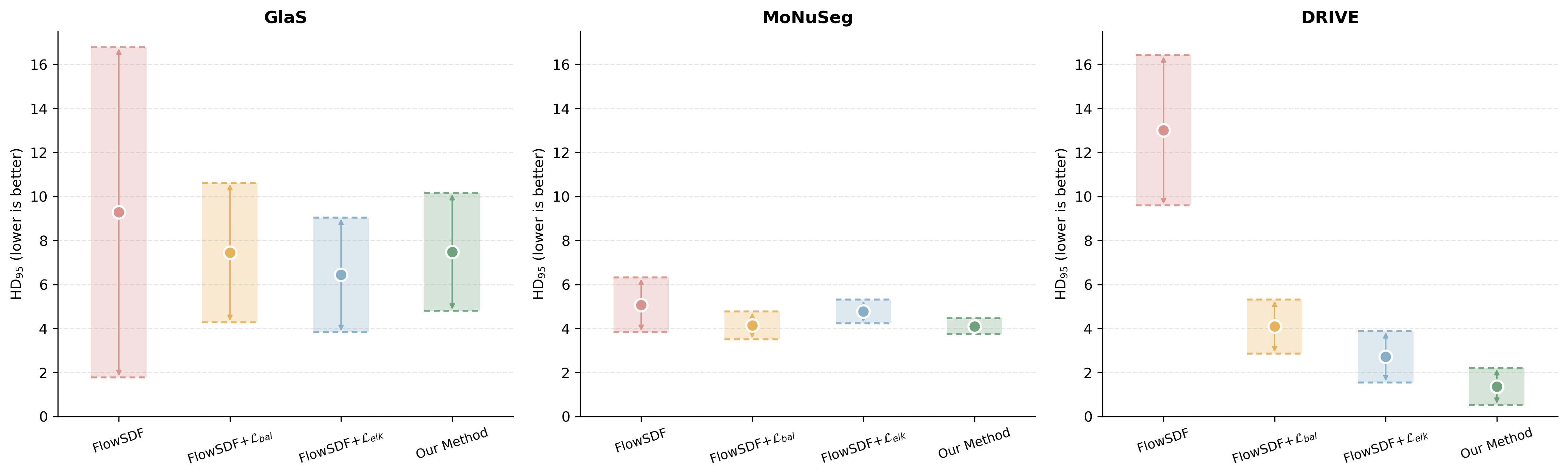}
    \caption{
    Comparison of the mean and standard deviation of HD$_{95}$ across five
    random seeds on GlaS, MoNuSeg, and DRIVE.
    The circular markers denote the mean values, while the vertical intervals
    indicate the corresponding mean $\pm$ standard deviation.
    Narrower intervals represent lower sensitivity to random initialization
    and more stable boundary localization.
    }
    \label{fig:hd95_ablation}
\end{figure}

Overall, introducing geometric constraints on top of the FlowSDF baseline
improves the region-overlap metrics on multiple datasets, reduces boundary
localization errors, and decreases the variation of several metrics across
different random seeds. This indicates that geometric regularization not only
improves the final segmentation accuracy but also enhances the stability of
the continuous SDF transport process.

The improvement is particularly evident on DRIVE, where vessel structures are
thin, elongated, and prone to discontinuities. The HD$_{95}$ of the FlowSDF
baseline is $13.0054\pm3.4181$, whereas the proposed full model reduces it to
$1.3652\pm0.8463$. In addition to the substantial reduction in the mean
boundary error, the smaller standard deviation demonstrates that the proposed
method produces more consistent predictions across different random seeds.
Without explicit geometric constraints, the continuous flow evolution is more
likely to suffer from local boundary drift or unstable zero-level-set
transport. By contrast, the full model better preserves vessel continuity and
boundary stability.

These results suggest that the proposed geometric constraints mainly improve boundary localization and reduce performance variation across repeated runs. Their effect on region overlap differs across datasets. The improvement is clearer on the DRIVE dataset, where preserving thin vessel structures is particularly important. This result suggests a balance between spatial smoothness and preserving fine details. Having analyzed the effects of the geometric constraints, we next examine their performance under different numbers of ODE function evaluations.

\subsection{Influence of NFE}

We vary only the number of Euler steps used to integrate the learned velocity field at inference, and report how this choice affects wall-clock inference time and segmentation accuracy on MoNuSeg, GlaS, and DRIVE.
At inference, the model integrates an ODE with an Euler solver from a noisy SDF initialization $s_0$ along the learned velocity field.
$N_{\mathrm{step}}$ is the number of discretization steps; a larger value typically increases the computational cost per forward pass.

\paragraph{Inference time}
Table~\ref{tab:ode_inference_time} reports the average inference time per test image (in seconds) for different $N_{\mathrm{step}}$ on the three datasets. Inference time increases monotonically with $N_{\mathrm{step}}$ and grows approximately linearly on MoNuSeg, GlaS, and DRIVE. Under the same hardware and implementation, our per-image inference time is substantially lower than diffusion-based segmenters such as SegDiff and MedSegDiff, which require tens to hundreds of sampling steps, and also lower than flow-matching baselines such as FlowSDF, which typically use more ODE steps.

\begin{table}[H]
  \centering
  \caption{Influence of ODE solver steps on inference time across diverse datasets.}
  \label{tab:ode_inference_time}
  \begin{tabular}{|c|c|c|c|}
    \hline
    $N_{\mathrm{step}}$ & MoNuSeg & GlaS & DRIVE \\
    \hline
    1   & 0.01 & 0.01 & 0.01 \\
    2   & 0.28 & 0.33 & 0.43 \\
    4   & 0.35 & 0.42 & 0.86 \\
    10  & 0.50 & 0.73 & 2.09 \\
    100 & 2.85 & 4.88 & 20.26 \\
    \hline
  \end{tabular}
\end{table}

\paragraph{Segmentation performance vs.\ $N_{\mathrm{step}}$}
Figure~\ref{fig:ode_nfe_metrics} plots segmentation metrics as a function of $N_{\mathrm{step}}$. Results show that the performance on all three datasets is already close to convergence at $N_{\mathrm{step}}{=}2$, where Dice and IoU have reached near-stable levels. When $N_{\mathrm{step}}$ is further increased, metrics do not consistently improve; some datasets exhibit minor fluctuations or slight degradation, suggesting that overly fine ODE discretization is unnecessary and may accumulate numerical integration error. Together with Fig.~\ref{fig:ode_nfe_metrics}, these results show that $N_{\mathrm{step}}{=}2$ provides a near-converged accuracy--efficiency setting, while further increases in $N_{\mathrm{step}}$ lead to slower and limited performance improvements.

To further examine sampling stability under limited function evaluations, we compare our method with FlowSDF under different NFE settings. This comparison evaluates how quickly each method reaches a stable segmentation state as the number of ODE integration steps increases. As shown in Table~\ref{tab:nfe_stability_flowsdf}, our method is expected to achieve stable segmentation performance with a smaller number of function evaluations, while FlowSDF generally requires more integration steps to approach a comparable stable state.

This behavior can be attributed to the high-order geometric constraints imposed on the recovered SDF. The biharmonic regularization suppresses local high-frequency oscillations in the implicit representation, while the Eikonal constraint maintains a stable distance-field structure near the zero level set. These constraints reduce the sensitivity of the learned flow to discretization errors during ODE integration, allowing the predicted SDF to evolve toward a geometrically stable state even under low NFE. In contrast, without explicit high-order geometric constraints, the flow evolution is more likely to exhibit boundary drift or unstable zero-level-set transport when the integration is coarse.

\begin{table}[H]
  \centering
  \caption{NFE stability comparison between FlowSDF and our method across datasets.}
  \label{tab:nfe_stability_flowsdf}
  \footnotesize
  \setlength{\tabcolsep}{5pt}
  \resizebox{\linewidth}{!}{%
  \begin{tabular}{|l|l|c|c|c|c|}
    \hline
    Dataset & Method & NFE & Dice$\uparrow$ & IoU/mIoU$\uparrow$ & HD$_{95}$$\downarrow$ \\
    \hline
    \multirow{8}{*}{GlaS}
    & \multirow{4}{*}{FlowSDF} & 1   & $0.5041\pm0.0673$ & $0.3397\pm0.0602$ & $28.8421\pm4.9108$ \\
    &                          & 2   & $0.9360\pm0.0295$ & $0.8811\pm0.0501$ & $9.6390\pm7.5978$ \\
    &                          & 4   & $0.9346\pm0.0342$ & $0.8791\pm0.0575$ & $9.2783\pm7.5068$ \\
    &                          & 100 & $0.9338\pm0.0367$ & $0.8780\pm0.0615$ & $10.0825\pm7.1377$ \\
    \cline{2-6}
    & \multirow{4}{*}{Ours}    & 1   & $0.4721\pm0.0280$ & $0.3094\pm0.0243$ & $33.4063\pm2.3655$ \\
    &                          & 2   & $\mathbf{0.9419\pm0.0077}$ & $\mathbf{0.8903\pm0.0138}$ & $\underline{7.4821\pm2.6826}$ \\
    &                          & 4   & $\underline{0.9412\pm0.0072}$ & $\underline{0.8890\pm0.0129}$ & $\mathbf{7.3909\pm2.7884}$ \\
    &                          & 100 & $0.9357\pm0.0115$ & $0.8793\pm0.0202$ & $7.7676\pm2.5543$ \\
    \hline
    \multirow{8}{*}{MoNuSeg}
    & \multirow{4}{*}{FlowSDF} & 1   & $0.2625\pm0.0037$ & $0.1511\pm0.0025$ & $11.2334\pm0.3013$ \\
    &                          & 2   & $0.8666\pm0.0321$ & $0.6226\pm0.0424$ & $5.0792\pm1.0939$ \\
    &                          & 4   & $\underline{0.8690\pm0.0286}$ & $0.6255\pm0.0379$ & $5.0725\pm1.2434$ \\
    &                          & 100 & $\mathbf{0.8744\pm0.0267}$ & $\underline{0.6327\pm0.0358}$ & $5.9615\pm1.1716$ \\
    \cline{2-6}
    & \multirow{4}{*}{Ours}    & 1   & $0.2571\pm0.0053$ & $0.1475\pm0.0035$ & $10.4009\pm0.1774$ \\
    &                          & 2   & $0.8623\pm0.0113$ & $\mathbf{0.6361\pm0.0148}$ & $\underline{4.0983\pm0.3682}$ \\
    &                          & 4   & $0.8588\pm0.0082$ & $0.6114\pm0.0106$ & $\mathbf{4.0881\pm0.1661}$ \\
    &                          & 100 & $0.8633\pm0.0074$ & $0.6173\pm0.0097$ & $5.0964\pm0.1567$ \\
    \hline
    \multirow{8}{*}{DRIVE}
    & \multirow{4}{*}{FlowSDF} & 1   & $0.1574\pm0.1495$ & $0.0873\pm0.1157$ & $75.4630\pm6.7892$ \\
    &                          & 2   & $0.7723\pm0.0220$ & $0.6296\pm0.0289$ & $19.8501\pm11.4841$ \\
    &                          & 4   & $0.7781\pm0.0226$ & $0.6374\pm0.0299$ & $13.0054\pm3.4181$ \\
    &                          & 100 & $0.7671\pm0.0325$ & $0.6281\pm0.0281$ & $13.1032\pm4.5125$ \\
    \cline{2-6}
    & \multirow{4}{*}{Ours}    & 1   & $0.1469\pm0.0049$ & $0.0793\pm0.0029$ & $71.0509\pm3.4470$ \\
    &                          & 2   & $0.9253\pm0.0208$ & $0.8616\pm0.0365$ & $\mathbf{1.3652\pm0.8463}$ \\
    &                          & 4   & $\mathbf{0.9369\pm0.0210}$ & $\mathbf{0.8821\pm0.0376}$ & $5.3652\pm3.8463$ \\
    &                          & 100 & $\underline{0.9365\pm0.0221}$ & $\underline{0.8813\pm0.0395}$ & $\underline{2.4691\pm2.4536}$ \\
    \hline
  \end{tabular}}
\end{table}

\begin{figure}[ht]
	\centering
	\subfigure{
		\begin{minipage}[h]{0.3\linewidth}
			\includegraphics[width=1\linewidth]{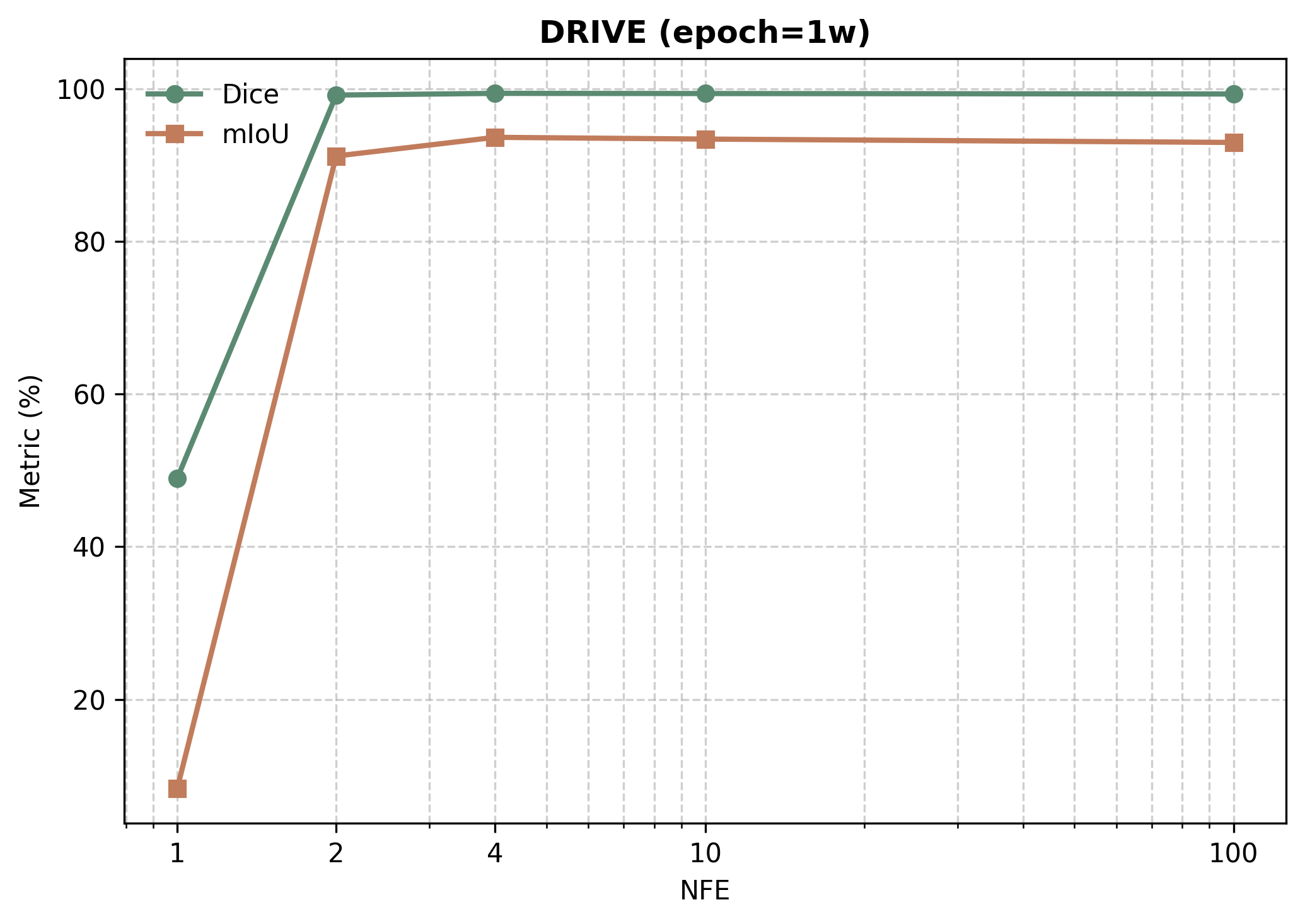} 
		\end{minipage}
		\label{Drive}
	}
    	\subfigure{
    		\begin{minipage}[h]{0.3\linewidth}
   		 	\includegraphics[width=1\linewidth]{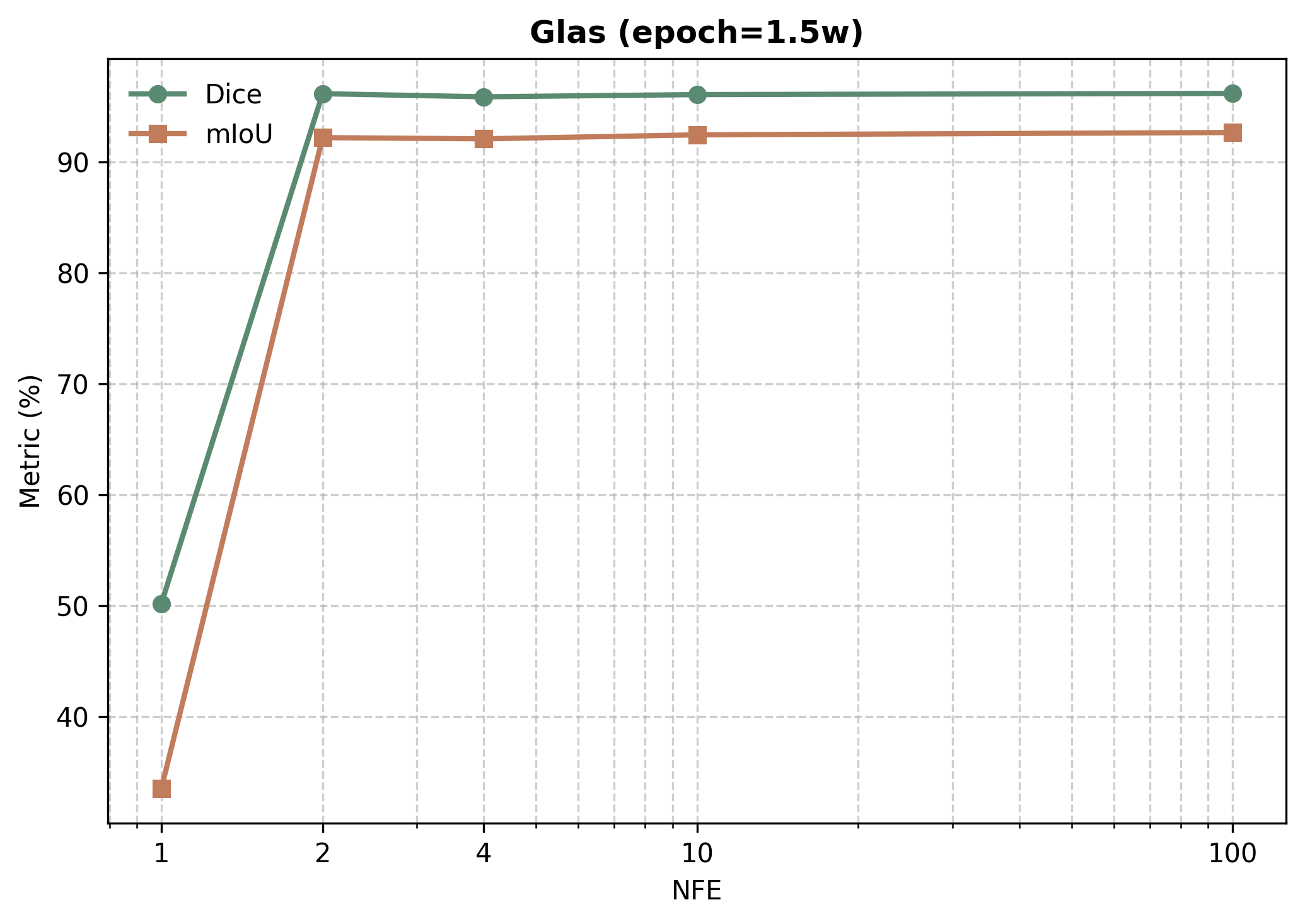}
    		\end{minipage}
		\label{Glas}
    	}
      \subfigure{
    		\begin{minipage}[h]{0.3\linewidth}
   		 	\includegraphics[width=1\linewidth]{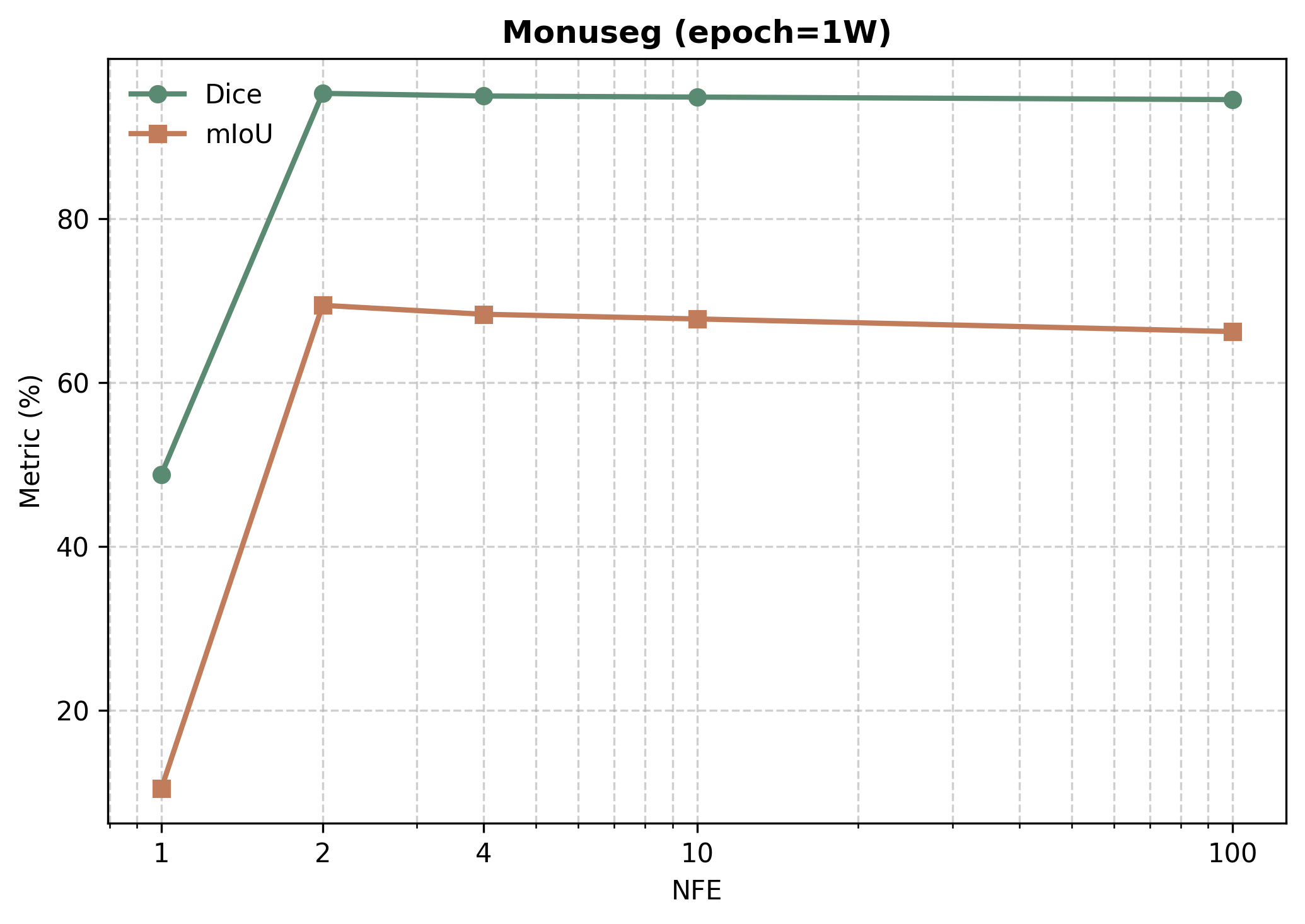}
    		\end{minipage}
		\label{MoNUseg}
    	}
		\caption{Influence of the number of ODE solver steps during sampling on experimental outcomes across diverse datasets.}
	\label{fig:ode_nfe_metrics}
\end{figure}

The main results in Table~\ref{tab:quantitative} use $N_{\mathrm{step}}{=}2$ at inference. This configuration allows the segmentation process to reach a near-converged state, while further integration steps produce only slower and limited performance improvements. It also requires fewer sampling steps and less runtime than generative baselines such as FlowSDF and SegDiff, highlighting the deployment efficiency of flow matching with SDF representations.

\subsection{Comparison with SOTA models}\label{subsec:comparison}
We compare our approach against representative baselines in medical image segmentation, including classical discriminative networks, Transformer architectures, diffusion/flow-matching generative models, and topology-preserving methods:
U-Net~\cite{ronneberger2015u}, U-Net++~\cite{zhou2018unet++}, TransUNet~\cite{chen2021transunet}, Swin-Unet~\cite{cao2022swin},
SegDiff~\cite{amit2021segdiff}, MedSegDiff~\cite{wu2024medsegdiff},
Topograph~\cite{lux2025topograph}, and FlowSDF~\cite{bogensperger2025flowsdf}.
U-Net and U-Net++ are the most widely used encoder--decoder baselines; TransUNet and Swin-Unet incorporate Transformers into a U-Net hybrid and a pure Transformer segmentation design, respectively;
SegDiff and MedSegDiff represent diffusion-based generative segmentation; Topograph enforces topological correctness via a graph-structured loss, which is particularly relevant for connected structures such as vessels and glands;
FlowSDF performs conditional flow matching in SDF space and serves as the closest generative counterpart to our method.
For a fair comparison, all methods use the same data splits and preprocessing as in \S\ref{subsec:datasets}; for SegDiff, MedSegDiff, FlowSDF, and our method, inference steps and sampling strategies follow the defaults in the original papers or our unified protocol.

Quantitative comparison with baselines is reported in Table~\ref{tab:quantitative}.
The results show that our method achieves competitive Dice, IoU , and HD$_{95}$ scores on GlaS, MoNuSeg, and DRIVE.
Compared with U-Net variants and Transformer baselines, generative modeling yields smoother boundaries and improved stability under limited training data.
Compared with diffusion models and flow-matching methods , our approach attains stable predictions in SDF space with very few ODE steps, enabling faster inference and a favorable accuracy--efficiency trade-off.

Figure~\ref{fig:qualitative} provides qualitative segmentation examples of our method on MoNuSeg and GlaS.
Each row corresponds to one dataset; from left to right: condition image $x$, SDF prediction $\tilde{m}_1$ obtained by ODE integration, overlay visualization $x\otimes m$, binary mask $m$ by zero-level thresholding of $\tilde{m}_1$, and expert annotation $m_{\mathrm{gt}}$.
The recovered structures are spatially coherent; thresholded $m$ closely matches $m_{\mathrm{gt}}$ on nuclei and gland contours, and the overlay column indicates good alignment with histological structures.
This illustrates that our flow-matching framework can generate clear segmentations from a noise prior to a final mask in SDF space.


\begin{table}[H]
  \centering
  \caption{Quantitative comparison on GlaS, MoNuSeg, and DRIVE.}
  \label{tab:quantitative}
  \footnotesize
  \setlength{\tabcolsep}{3pt}
  \resizebox{\linewidth}{!}{%
  \begin{tabular}{|c|c|c|c|c|c|c|c|c|c|c|}
    \hline
    \multirow{2}{*}{Method} & \multirow{2}{*}{Year} &
    \multicolumn{3}{|c|}{GlaS} & \multicolumn{3}{c|}{MoNuSeg} & \multicolumn{3}{c|}{DRIVE} \\
    \cline{3-5}\cline{6-8}\cline{9-11}
    & &
    IoU$\uparrow$ & Dice$\uparrow$ & HD$_{95}$$\downarrow$ &
    IoU$\uparrow$ & Dice$\uparrow$ & HD$_{95}$$\downarrow$ &
    IoU$\uparrow$ & Dice$\uparrow$ & HD$_{95}$$\downarrow$ \\
    \hline
    U-net       & 2015 & 0.6800 & 0.7341 & 28.96 & 0.5791 & 0.7334 & 16.71 & 0.6067 & 0.7552 & 37.96 \\
    U-net++     & 2018 & 0.6555 & 0.7803 & 24.71 & \textbf{0.6604} & 0.7949 & 14.65 & 0.6685 & \underline{0.8778} & 35.75 \\
    TransUNet   & 2021 & 0.8225 & \underline{0.8978} & 15.96 & 0.5267 & 0.6872 & 11.35 & 0.6164 & 0.7627 & 28.28 \\
    Swin-Unet   & 2021 & 0.8194 & 0.8706 & 14.39 & 0.5331 & 0.5970 & 11.27 & 0.6455 & 0.7204 & 29.48 \\
    SegDiff     & 2021 & 0.6387 & 0.7795 & 24.72 & 0.6259 & 0.7688 & 12.78 & 0.2573 & 0.2719 & 107.02\\
    MedSegDiff  & 2022 & 0.6479 & 0.7185 & 21.57 & \underline{0.6378} & 0.7869 & 13.42 & 0.4593 & 0.6295 & 28.90 \\
    Topograph   & 2025 & 0.7878 & 0.8807 & 20.67 & 0.6186 & 0.7643 & 7.57  & \underline{0.7123} & 0.8315 & \underline{5.76} \\
    FlowSDF     & 2025 & \underline{0.8771} & 0.8931 & \underline{6.61}  & 0.6071 & \underline{0.9175} & \underline{5.45}  & 0.6611 & 0.7960 & 15.52 \\
    Our Method  & 2026 & \textbf{0.8915} & \textbf{0.9427} & \textbf{5.94}  & 0.6314 & \textbf{0.9222} & \textbf{4.98}  & \textbf{0.9149} & \textbf{0.9555} & \textbf{1.00} \\

    \hline
  \end{tabular}}
\end{table}

For FlowSDF and our method, the results reported in Table~\ref{tab:quantitative} are selected as the best-performing runs among different random seeds. The mean and standard deviation over multiple random seeds are further reported in the subsequent ablation study to provide a more complete assessment of stability.

\begin{figure}[h]
    \centering
    
    \begin{minipage}[c][0.145\textheight][c]{0.12\linewidth}
        \centering
        \textbf{DRIVE}
    \end{minipage}%
    \subfigure{%
        \begin{minipage}[c]{0.145\linewidth}
            \includegraphics[width=\linewidth]{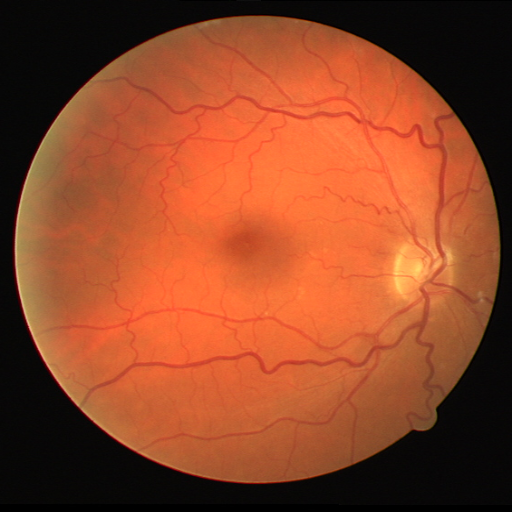}
        \end{minipage}
        \label{fig:drive1}
    }%
    \subfigure{%
        \begin{minipage}[c]{0.145\linewidth}
            \includegraphics[width=\linewidth]{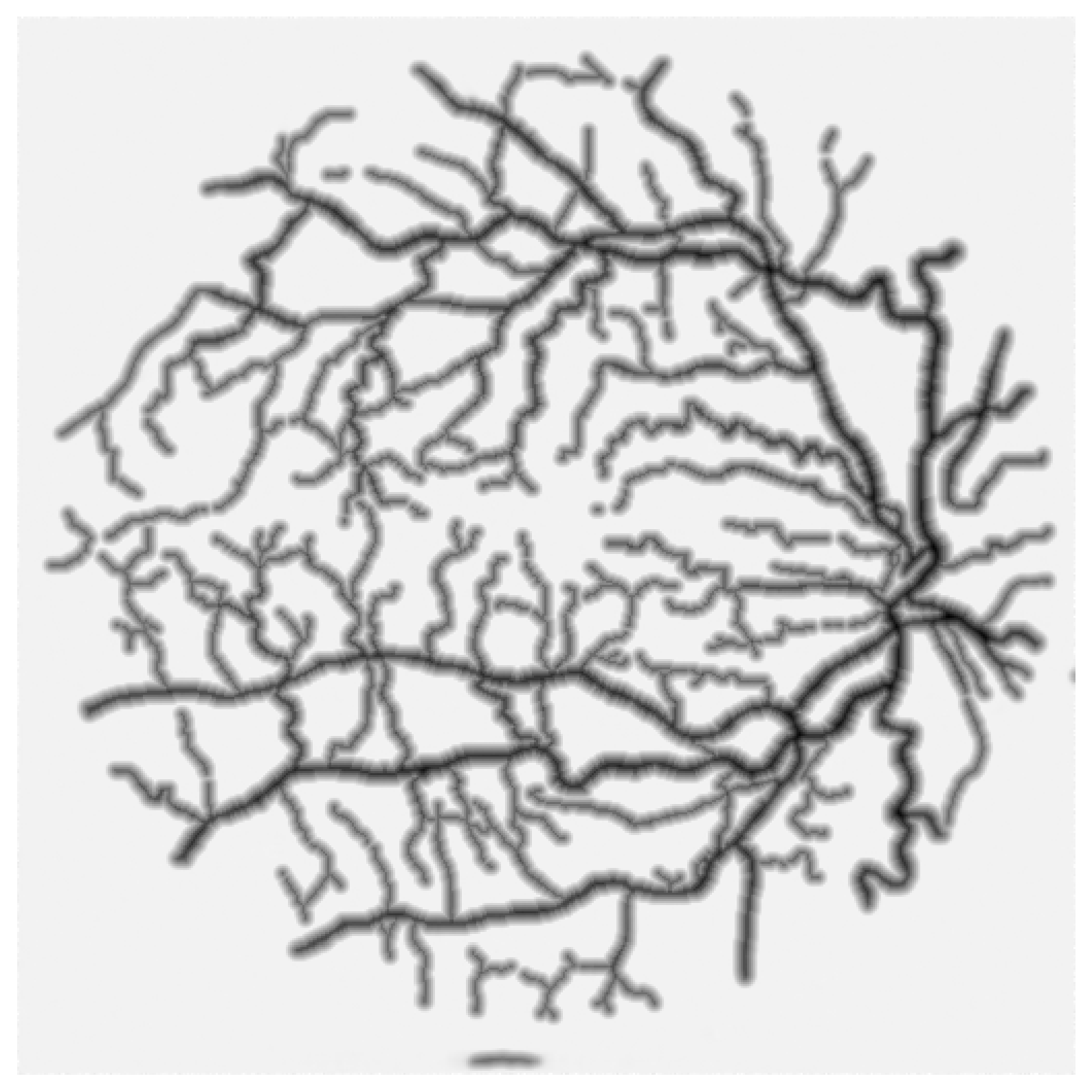}
        \end{minipage}
        \label{fig:drive2}
    }%
    \subfigure{%
        \begin{minipage}[c]{0.145\linewidth}
            \includegraphics[width=\linewidth]{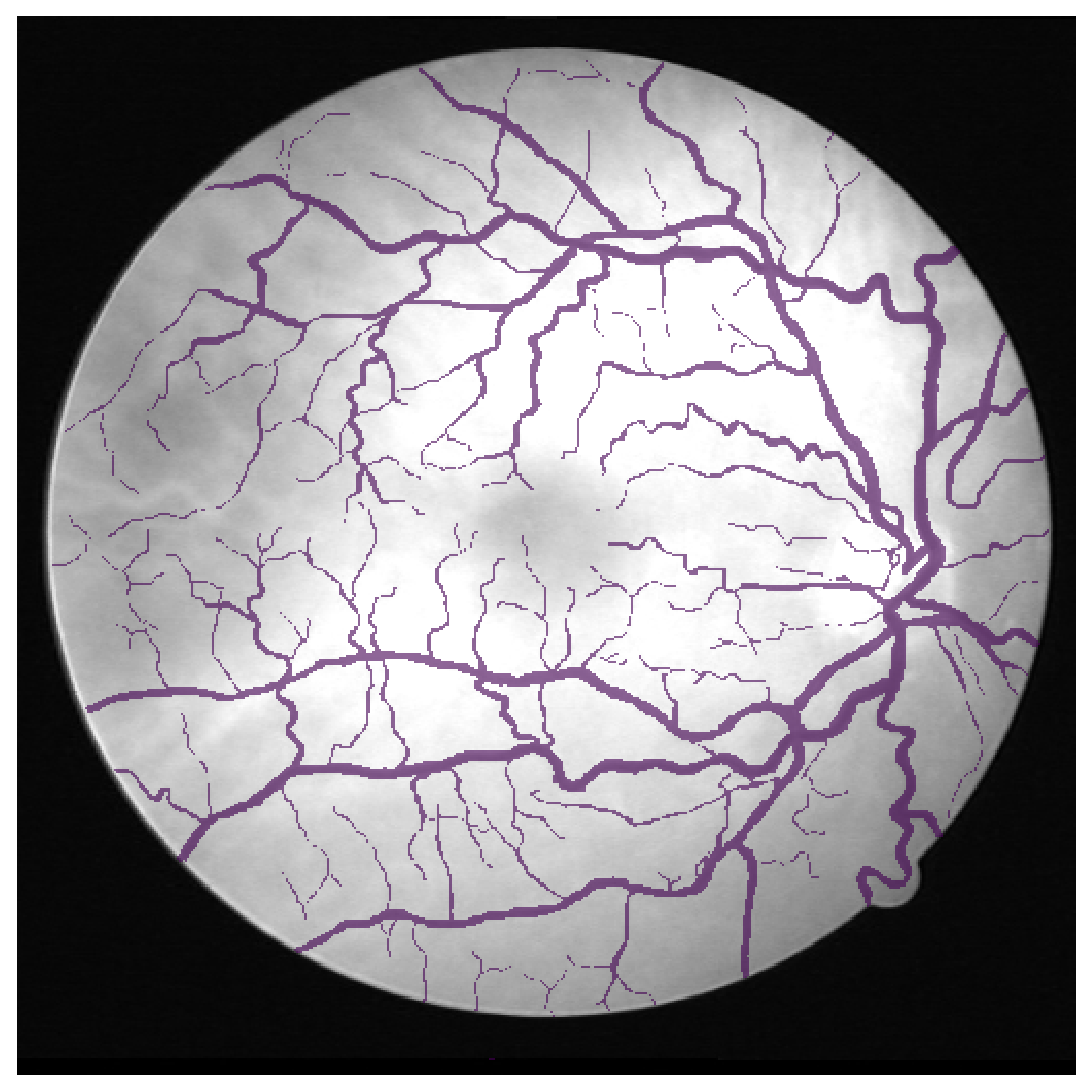}
        \end{minipage}
        \label{fig:drive3}
    }%
    \subfigure{%
        \begin{minipage}[c]{0.145\linewidth}
            \includegraphics[width=\linewidth]{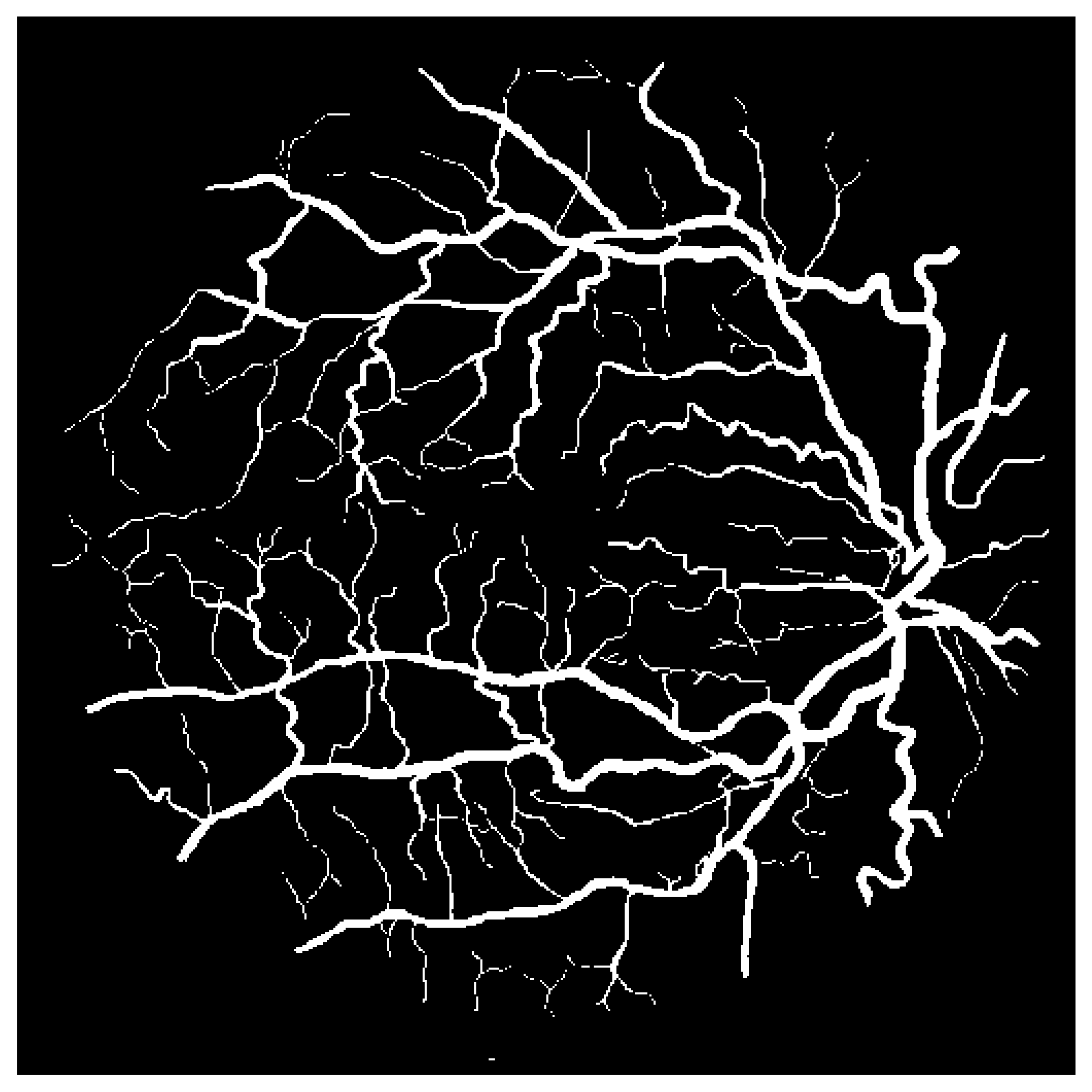}
        \end{minipage}
        \label{fig:drive4}
    }%
    \subfigure{%
        \begin{minipage}[c]{0.145\linewidth}
            \includegraphics[width=\linewidth]{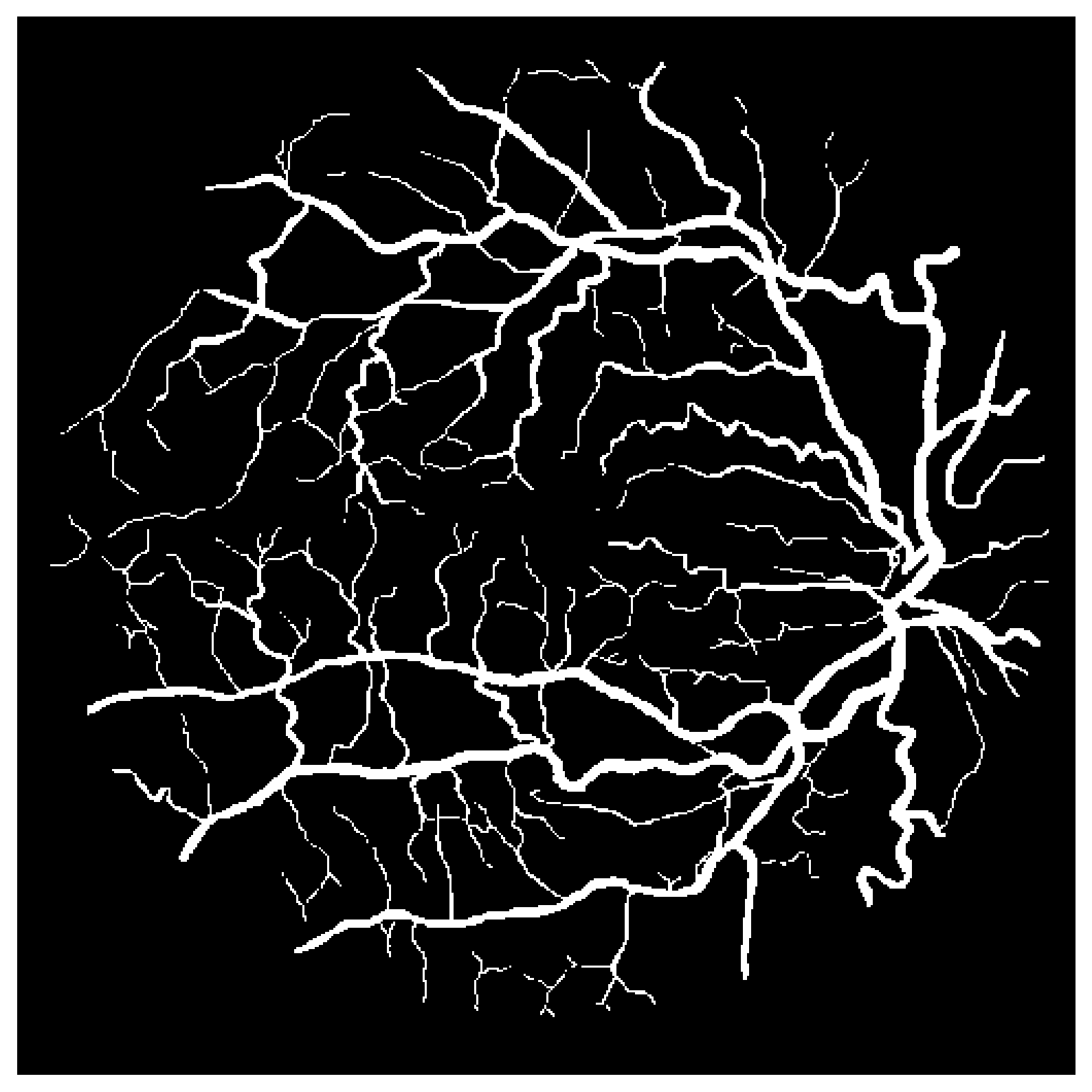}
        \end{minipage}
        \label{fig:drive5}
    }

    \vspace{0.5em}

    \begin{minipage}[c][0.12\textheight][c]{0.12\linewidth}
        \centering
        \textbf{GlaS}
    \end{minipage}%
    \subfigure{%
        \begin{minipage}[c]{0.145\linewidth}
            \includegraphics[width=\linewidth]{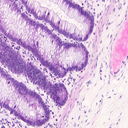}
        \end{minipage}
        \label{fig:glas1}
    }%
    \subfigure{%
        \begin{minipage}[c]{0.145\linewidth}
            \includegraphics[width=\linewidth]{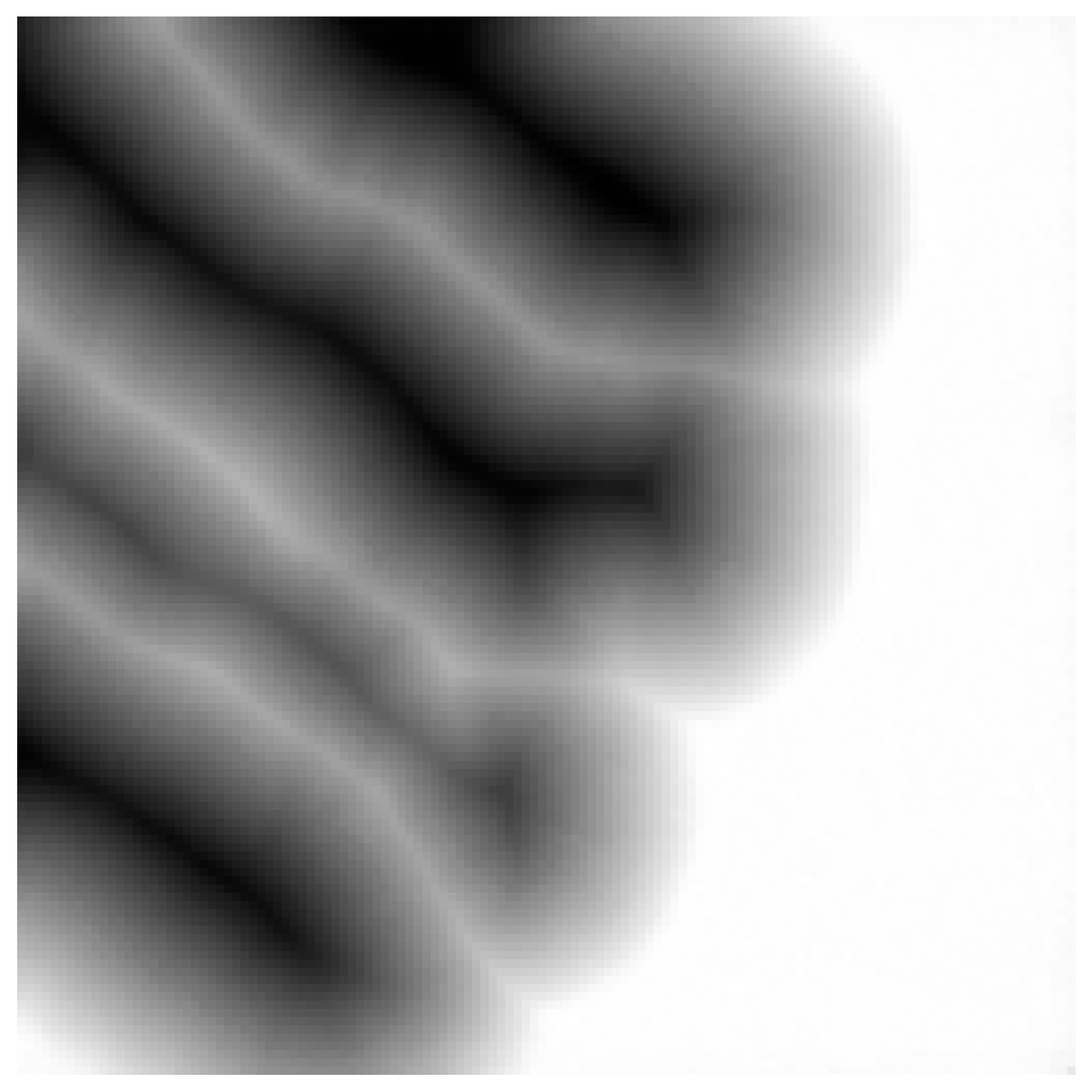}
        \end{minipage}
        \label{fig:glas2}
    }%
    \subfigure{%
        \begin{minipage}[c]{0.145\linewidth}
            \includegraphics[width=\linewidth]{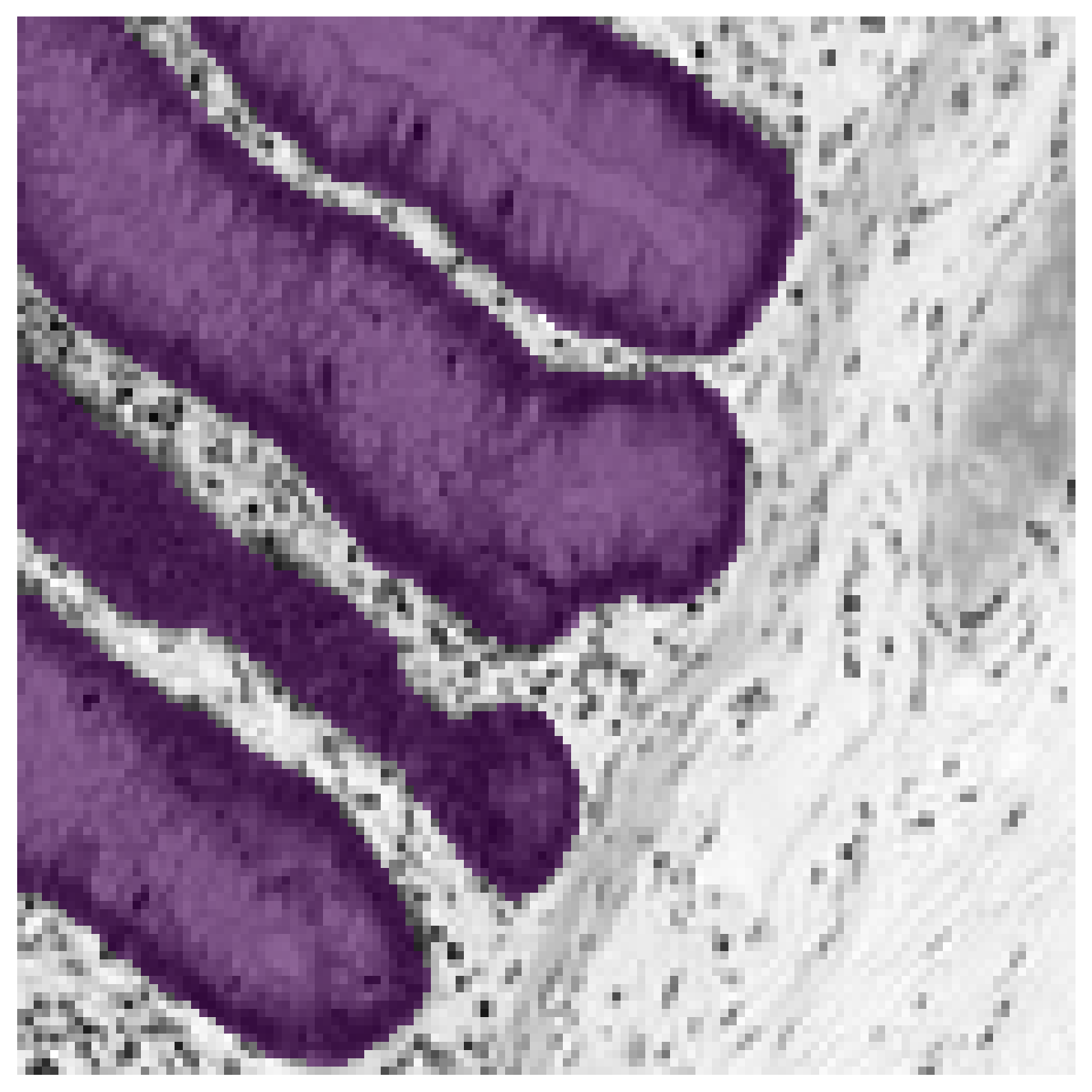}
        \end{minipage}
        \label{fig:glas3}
    }%
    \subfigure{%
        \begin{minipage}[c]{0.145\linewidth}
            \includegraphics[width=\linewidth]{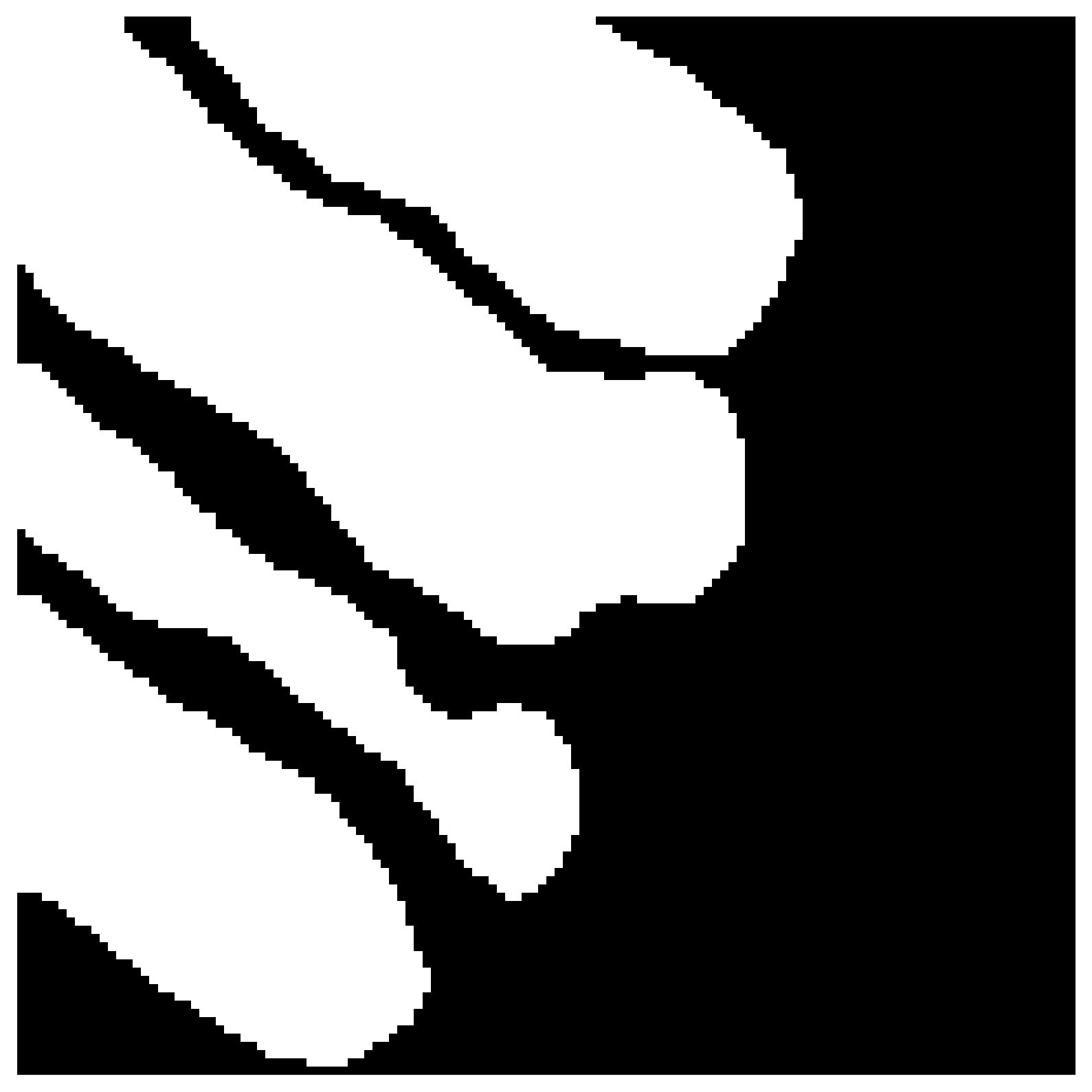}
        \end{minipage}
        \label{fig:glas4}
    }%
    \subfigure{%
        \begin{minipage}[c]{0.145\linewidth}
            \includegraphics[width=\linewidth]{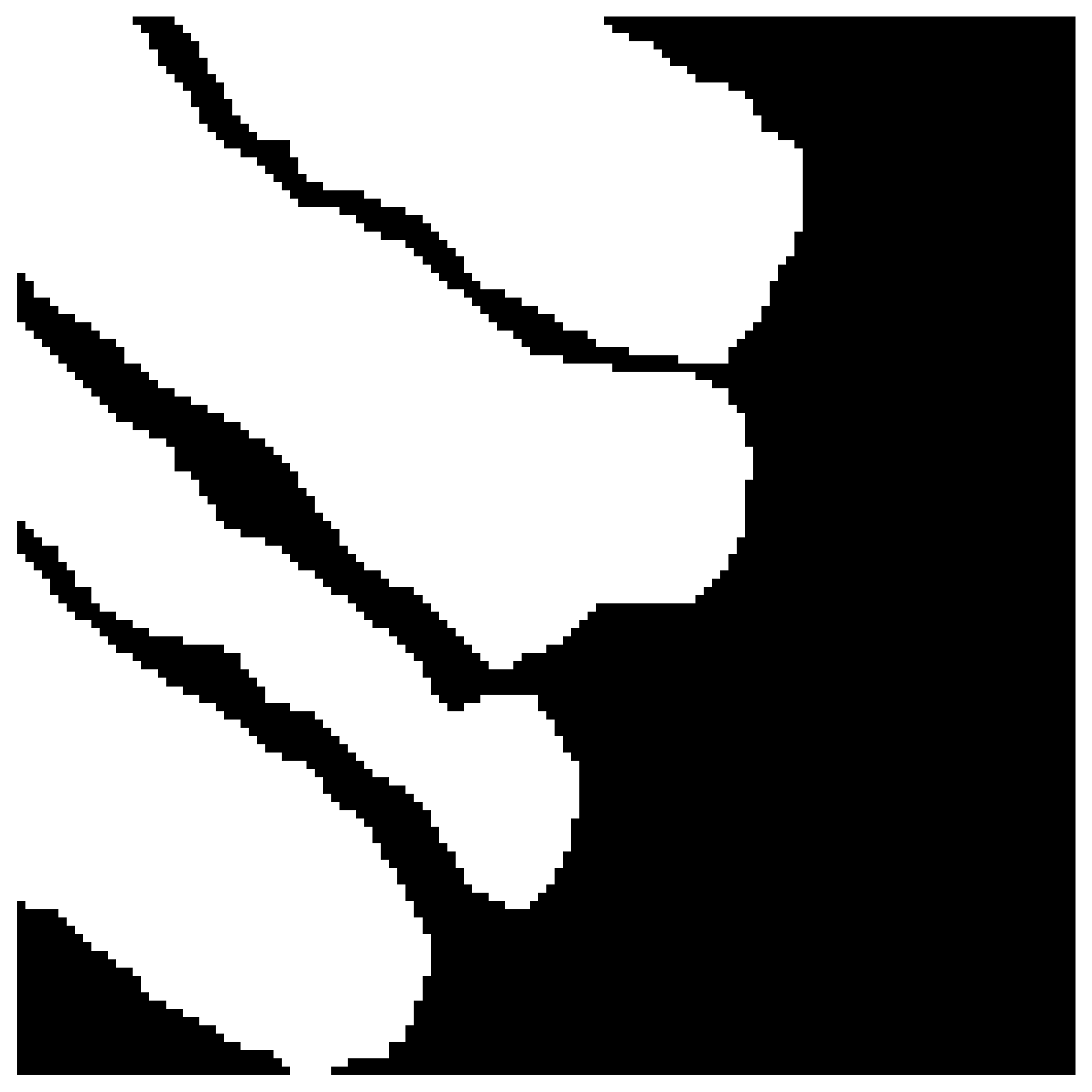}
        \end{minipage}
        \label{fig:glas5}
    }

    \vspace{0.5em}

    \begin{minipage}[c][0.12\textheight][c]{0.12\linewidth}
        \centering
        \textbf{\small MonuSeg}
    \end{minipage}%
    \subfigure{%
        \begin{minipage}[c]{0.145\linewidth}
            \includegraphics[width=\linewidth]{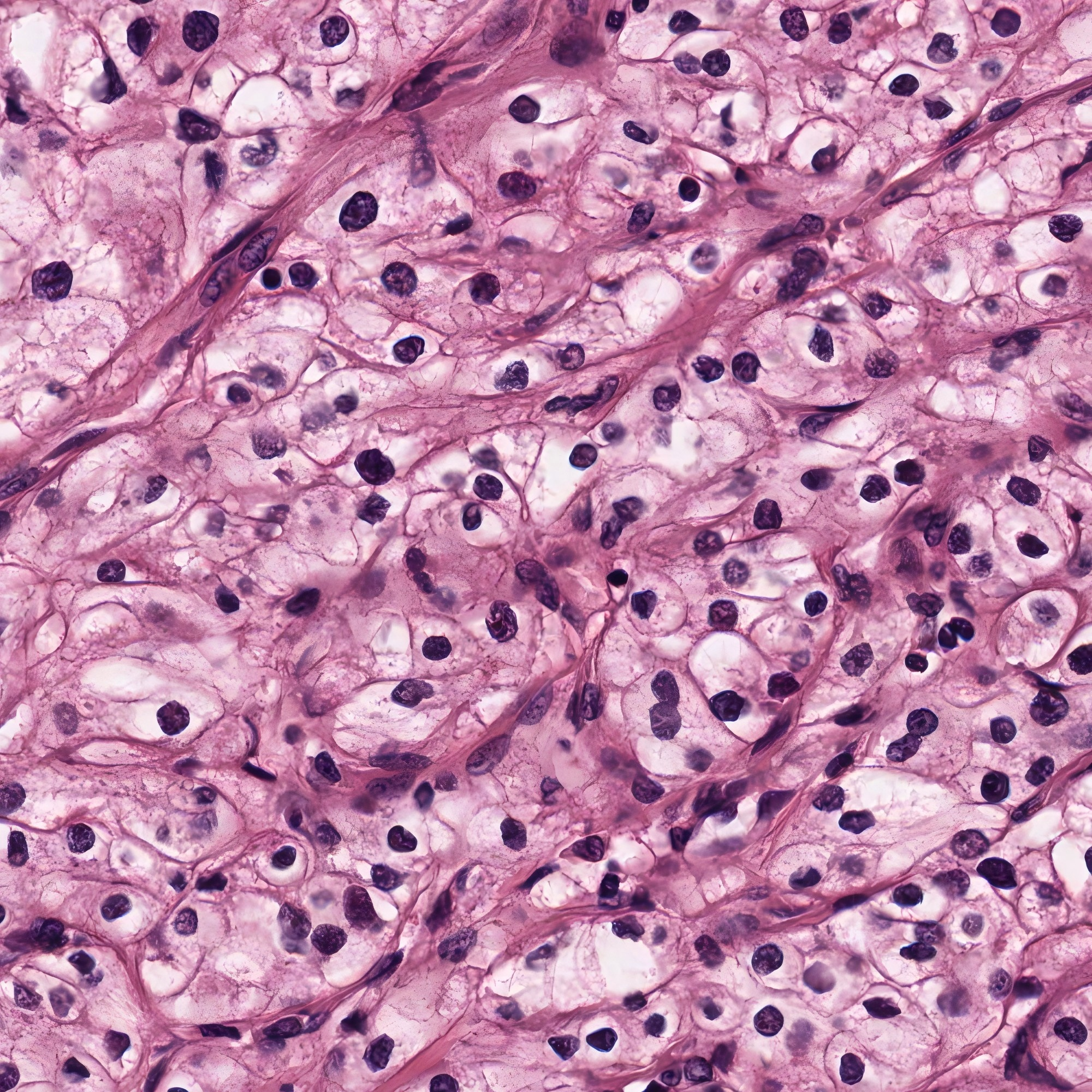}
        \end{minipage}
        \label{fig:monu1}
    }%
    \subfigure{%
        \begin{minipage}[c]{0.145\linewidth}
            \includegraphics[width=\linewidth]{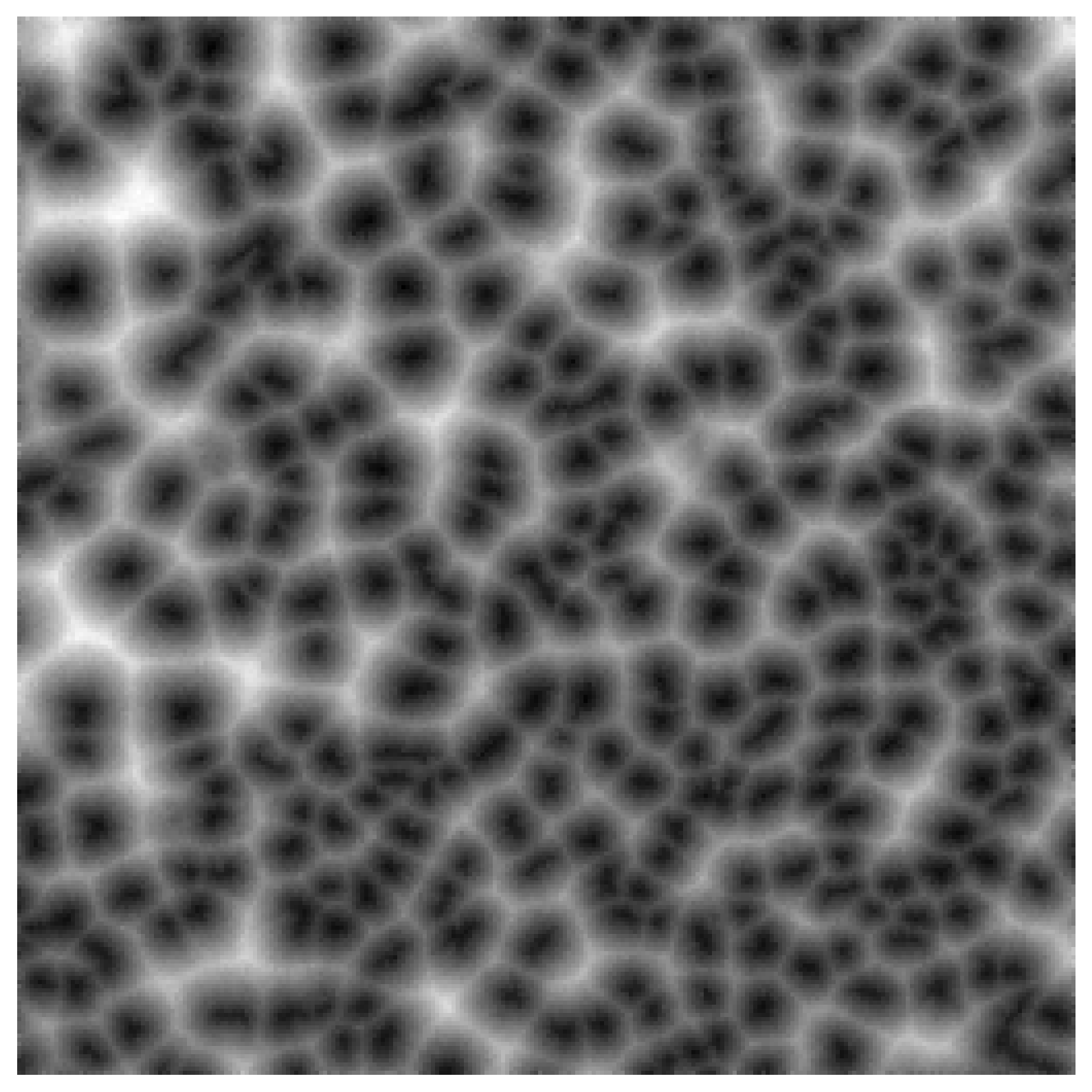}
        \end{minipage}
        \label{fig:monu2}
    }%
    \subfigure{%
        \begin{minipage}[c]{0.145\linewidth}
            \includegraphics[width=\linewidth]{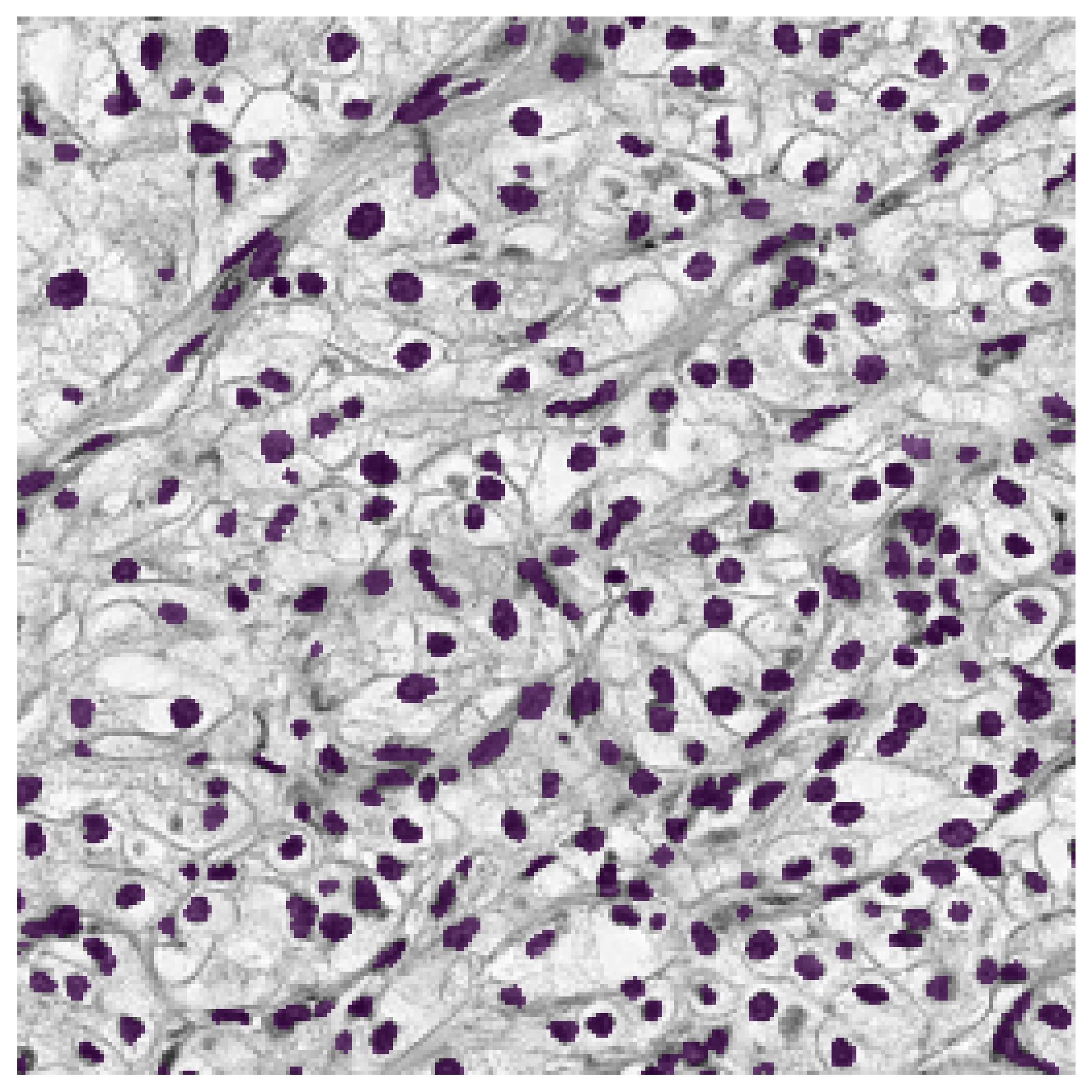}
        \end{minipage}
        \label{fig:monu3}
    }%
    \subfigure{%
        \begin{minipage}[c]{0.145\linewidth}
            \includegraphics[width=\linewidth]{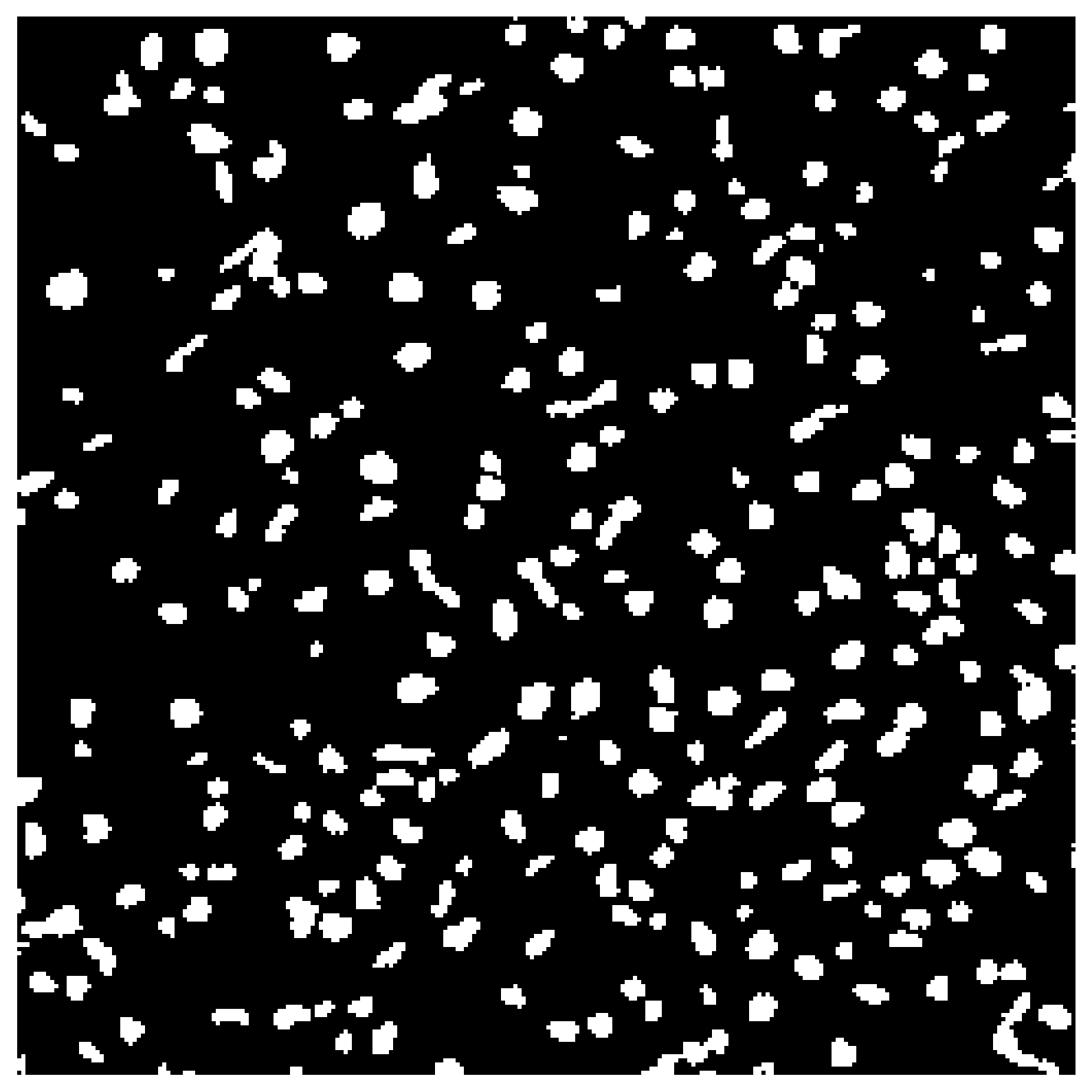}
        \end{minipage}
        \label{fig:monu4}
    }%
    \subfigure{%
        \begin{minipage}[c]{0.145\linewidth}
            \includegraphics[width=\linewidth]{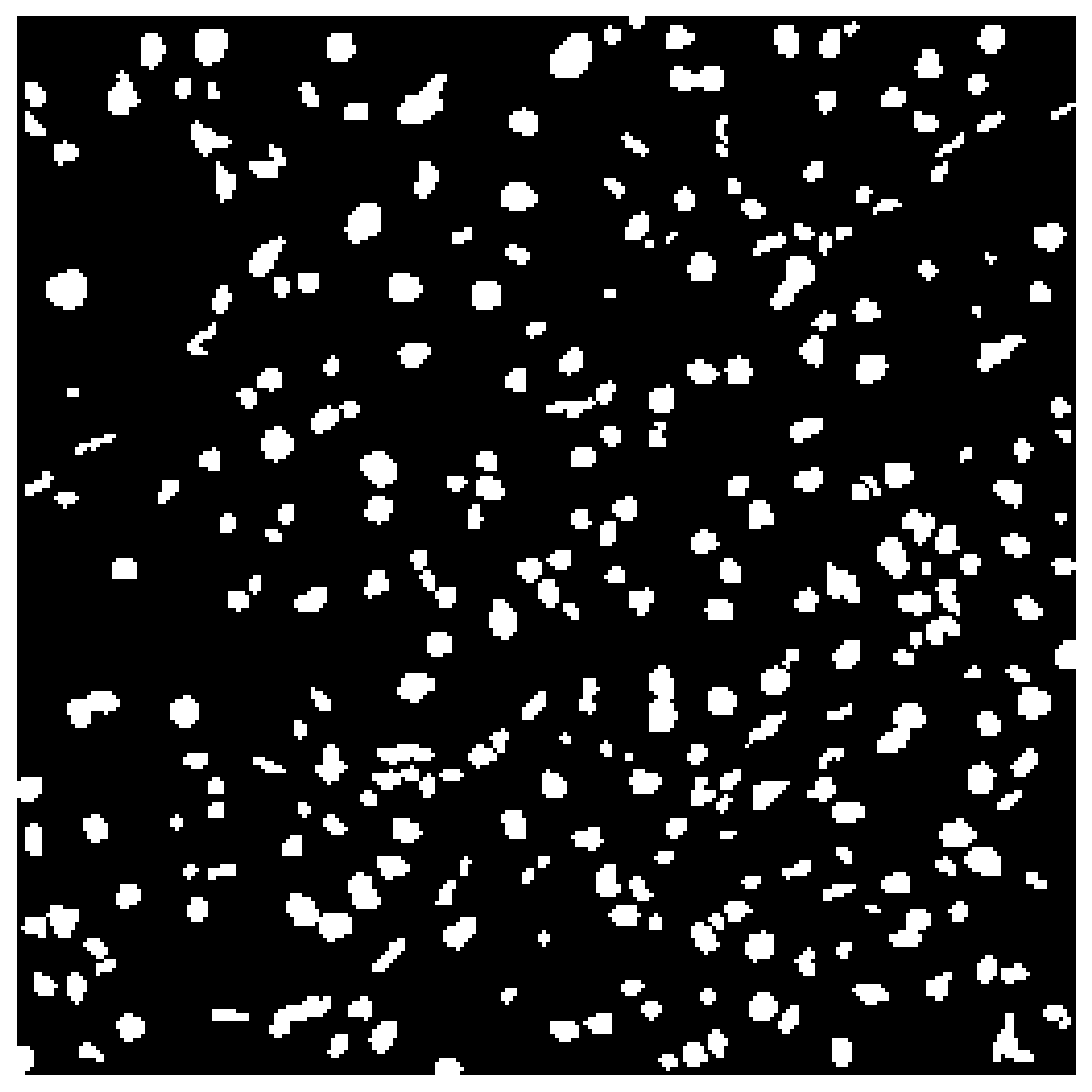}
        \end{minipage}
        \label{fig:monu5}
    }

    \caption{Qualitative results on DRIVE, GlaS, and MonuSeg datasets. 
    From left to right: Original Image, Predicted SDF Mask, Thresholded Prediction Overlay, Thresholded Prediction, and Ground Truth.}
    
    \label{fig:qualitative}
\end{figure}

To further examine the boundary localization capability of different methods on thin and low-contrast vessel structures, Figure~\ref{fig:drive-segmentation-comparison} presents a detailed qualitative comparison on a representative DRIVE image. In addition to the direct segmentation predictions, overlay images are included to show the alignment between the predicted vessels and the retinal image, while pixel-wise error maps explicitly distinguish false-positive and false-negative regions from correctly classified vessel and background pixels.

\begin{figure}[H]
    \centering
    \captionsetup{font=small,labelfont=bf}
    \captionsetup[subfigure]{font=small,labelfont=bf,skip=3pt}

    \newcommand{\resultsubfigure}[5]{%
        \begin{subfigure}[t]{0.188\textwidth}
            \centering
            \includegraphics[width=\linewidth,keepaspectratio]
            {#2}\par\vspace{2pt}
            \includegraphics[width=\linewidth,keepaspectratio]
            {#3}\par\vspace{2pt}
            \includegraphics[width=\linewidth,keepaspectratio]
            {#4}
            \caption{#1}
            \label{#5}
        \end{subfigure}%
    }

    \vspace{0.55em}

    \resultsubfigure
        {GT}
        {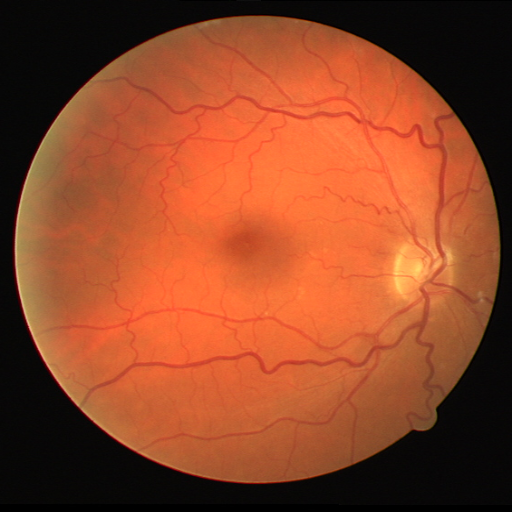}
        {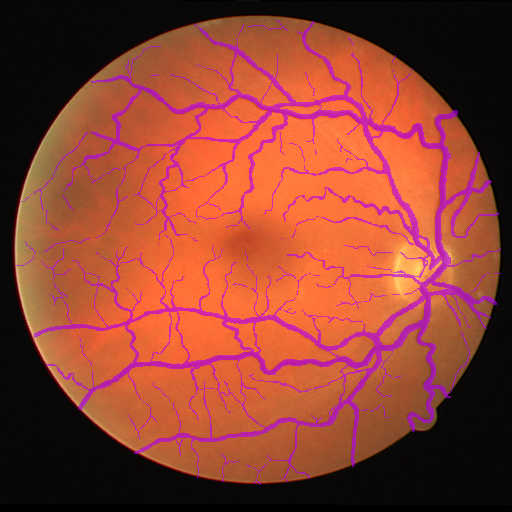}
        {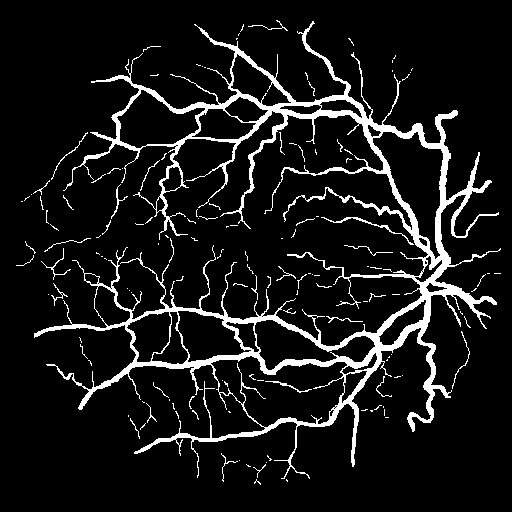}
        {fig:gt}
    \hfill
    \resultsubfigure
        {U-net}
        {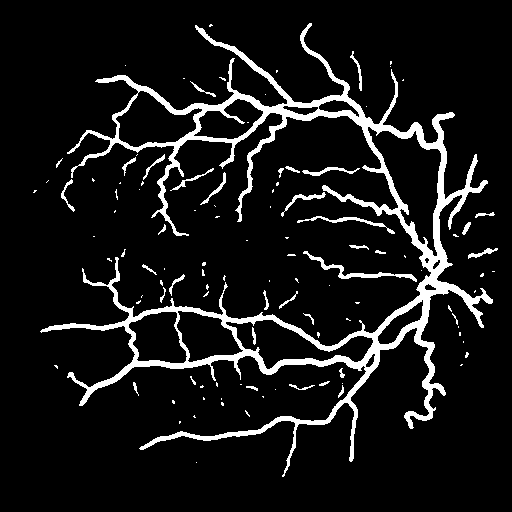}
        {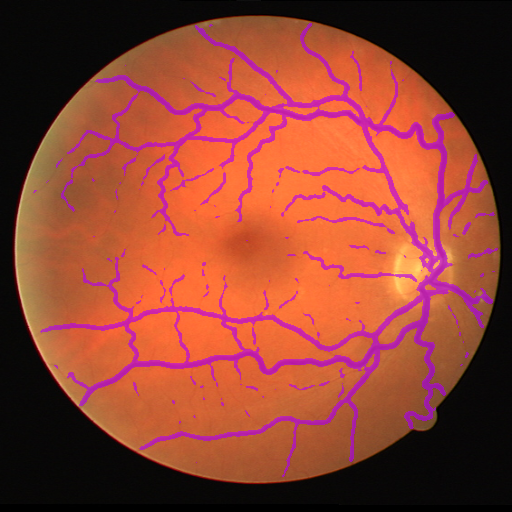}
        {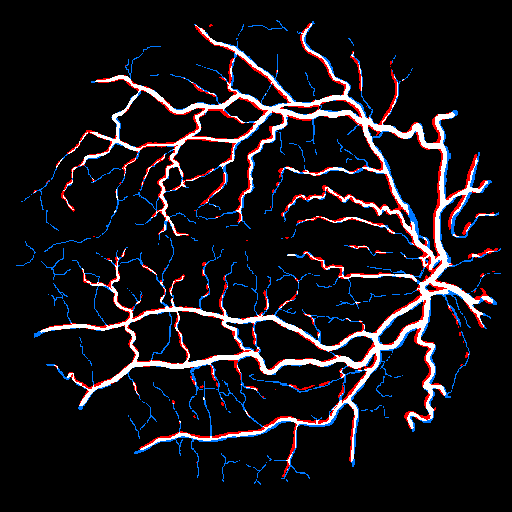}
        {fig:unet}
    \hfill
    \resultsubfigure
        {U-net++}
        {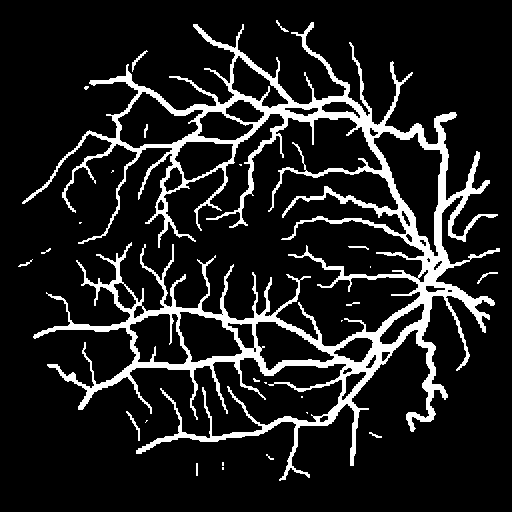}
        {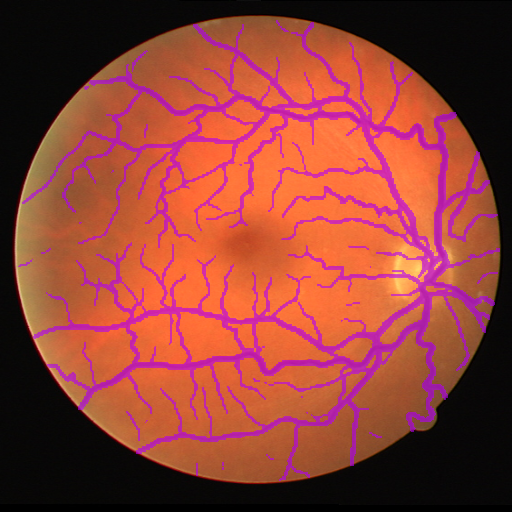}
        {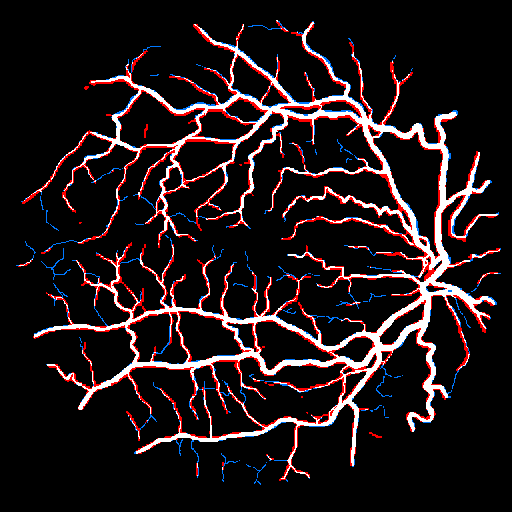}
        {fig:unetpp}
    \hfill
    \resultsubfigure
        {TransUNet}
        {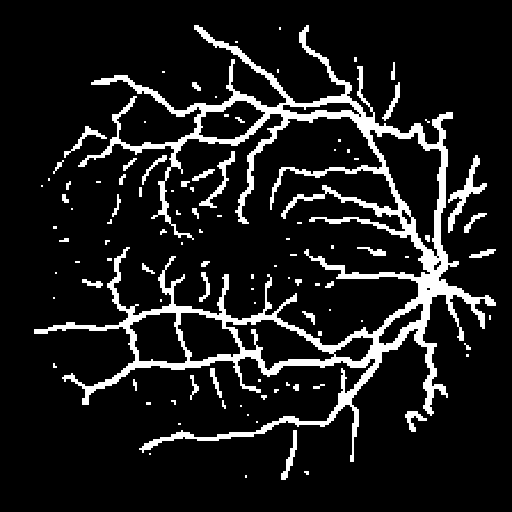}
        {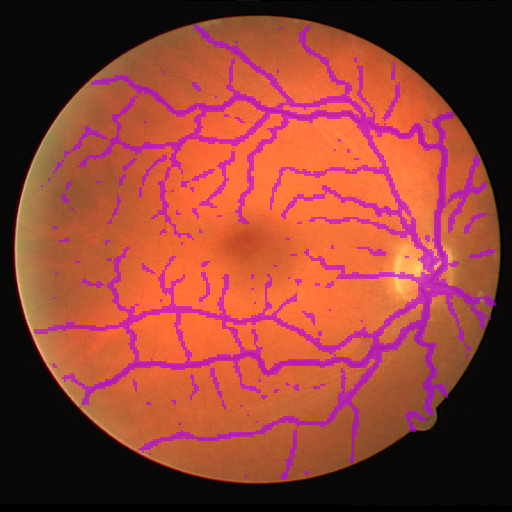}
        {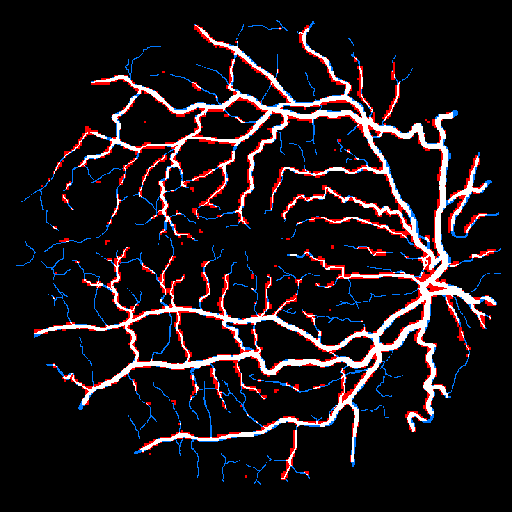}
        {fig:transunet}
    \hfill
    \resultsubfigure
        {Swin-Unet}
        {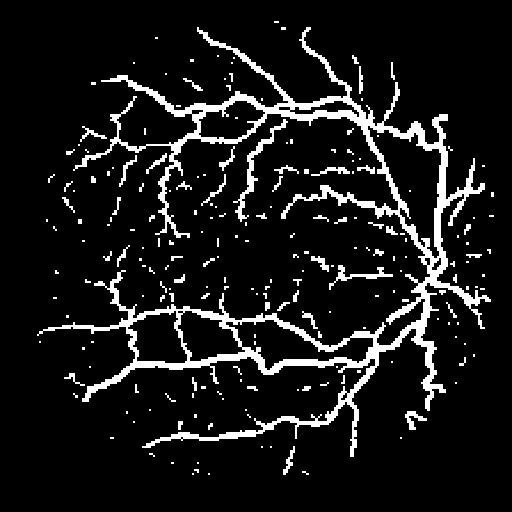}
        {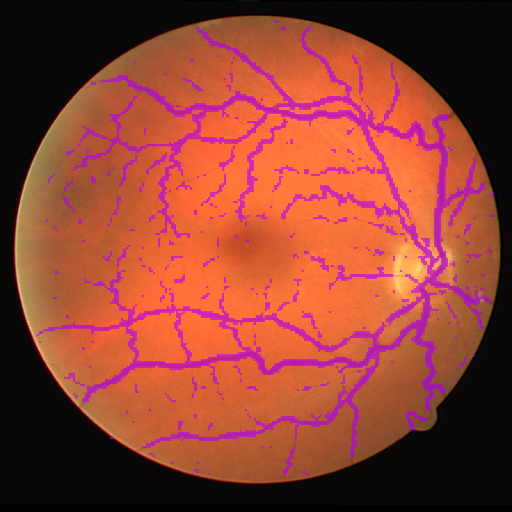}
        {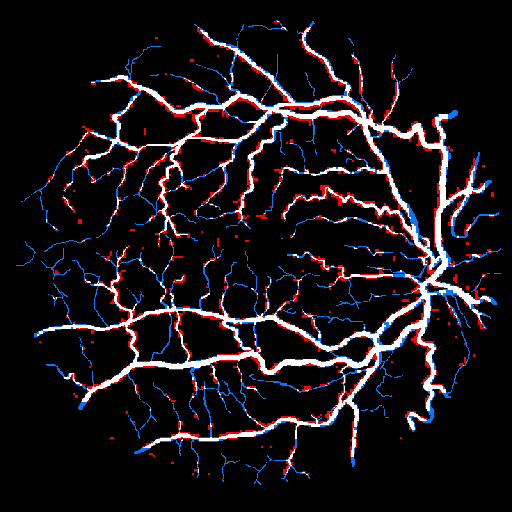}
        {fig:swinunet}

    \par\vspace{0.85em}

    \resultsubfigure
        {SegDiff}
        {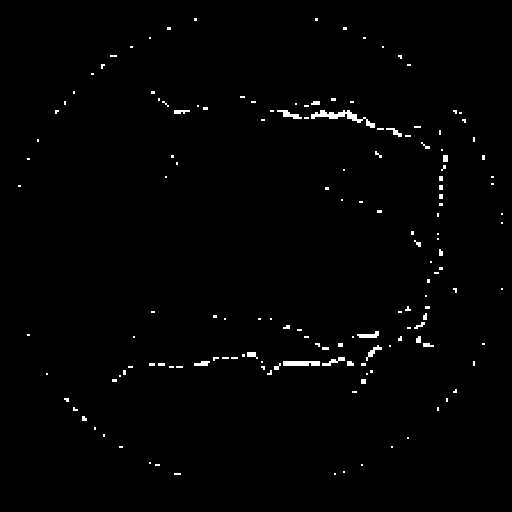}
        {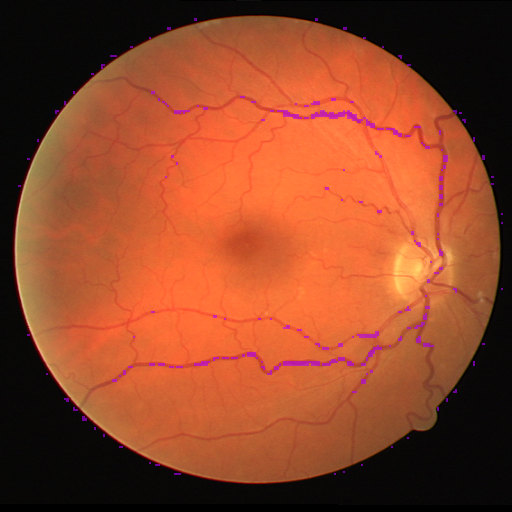}
        {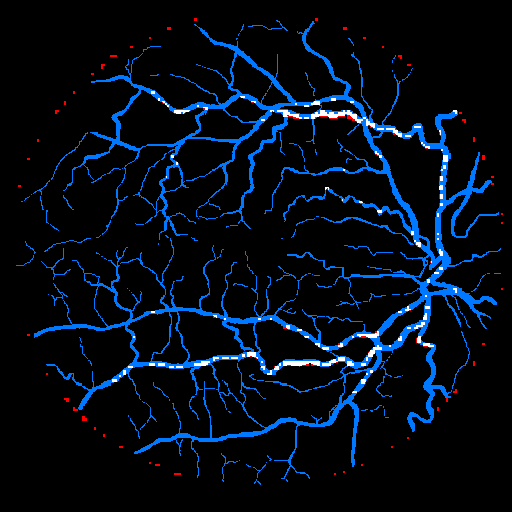}
        {fig:segdiff}
    \hfill
    \resultsubfigure
        {MedSegDiff}
        {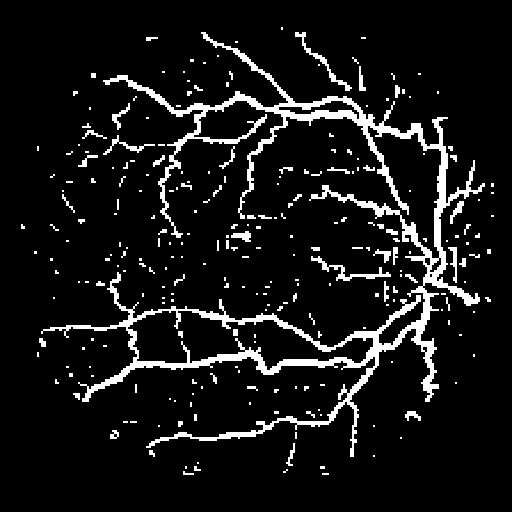}
        {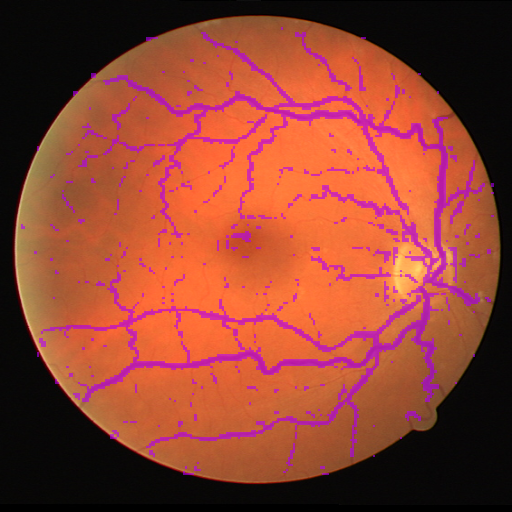}
        {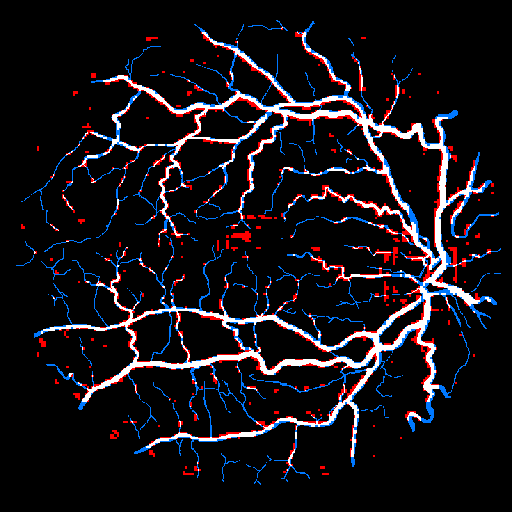}
        {fig:medsegdiff}
    \hfill
    \resultsubfigure
        {Topograph}
        {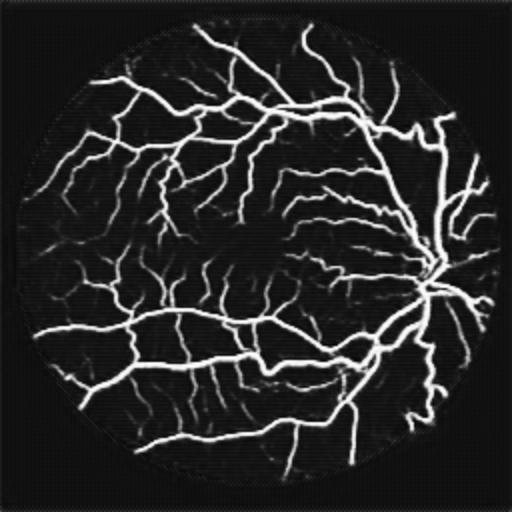}
        {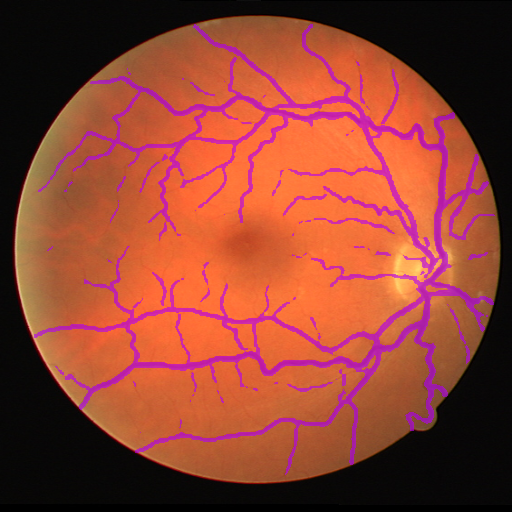}
        {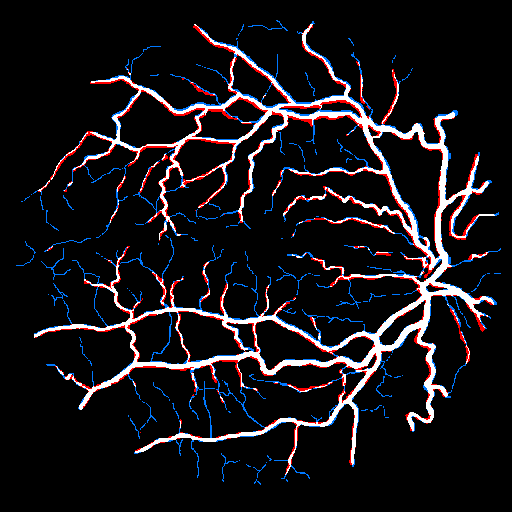}
        {fig:topograph}
    \hfill
    \resultsubfigure
        {FlowSDF}
        {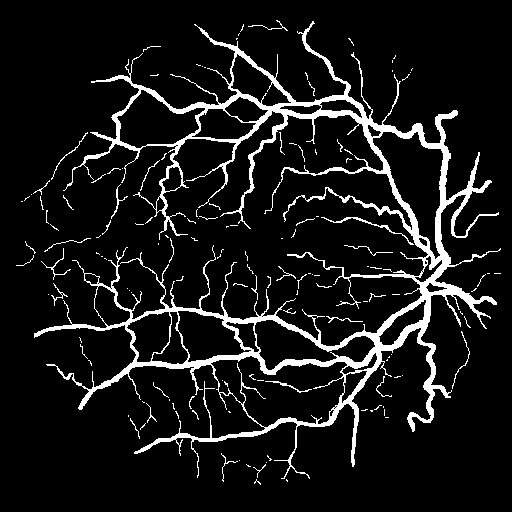}
        {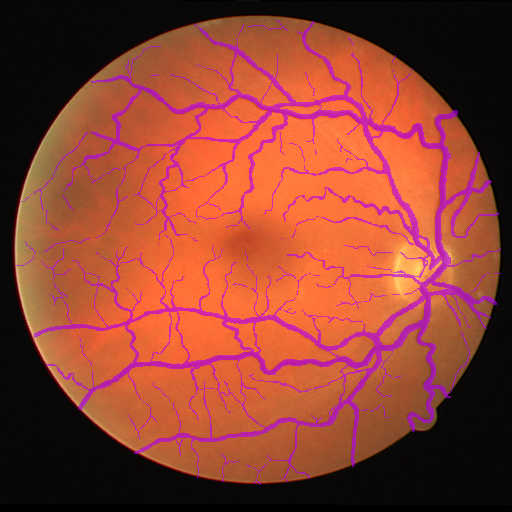}
        {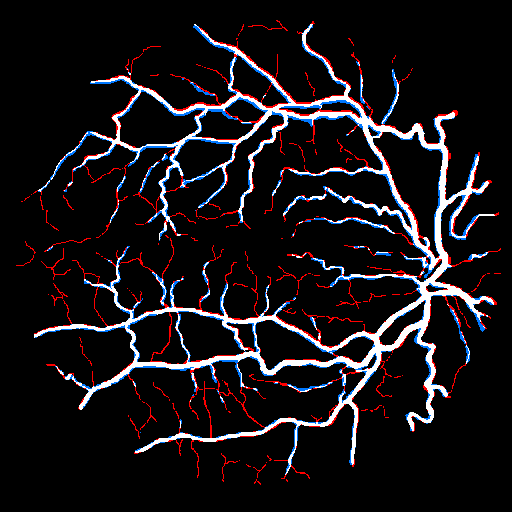}
        {fig:flowsdf}
    \hfill
    \resultsubfigure
        {Our Method}
        {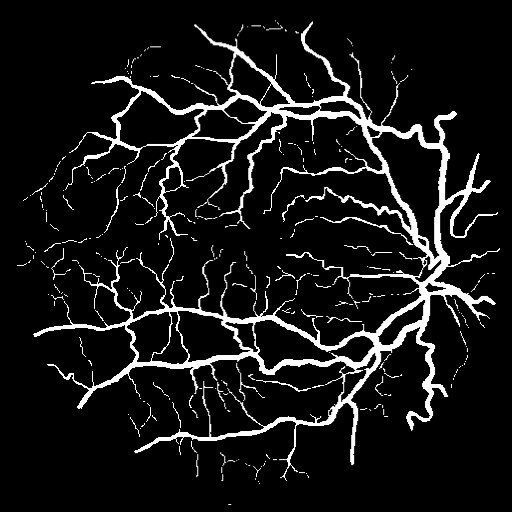}
        {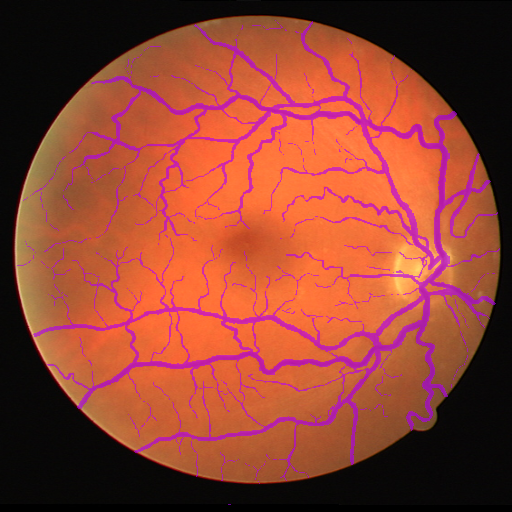}
        {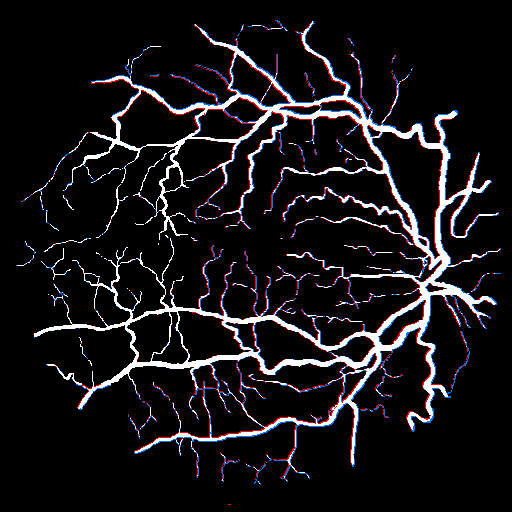}
        {fig:ours}

    \caption{%
        Qualitative comparison of retinal vessel segmentation on DRIVE, including predicted masks, image overlays, and pixel-wise error maps.
    }
    \label{fig:drive-segmentation-comparison}
\end{figure}

As shown in Figure~\ref{fig:drive-segmentation-comparison}, ten subfigures are arranged in two rows. Subfigure~(a) presents the reference images, including the original retinal image, the ground-truth overlay, and the ground-truth vessel mask from top to bottom. Subfigures~(b)-(j) show the results obtained by U-net, U-net++, TransUNet, Swin-Unet, SegDiff, MedSegDiff, Topograph, FlowSDF, and the proposed method, respectively. Within each method subfigure, the three images from top to bottom correspond to the predicted vessel mask, the prediction overlaid on the original image, and the pixel-wise error map relative to the ground truth. Purple denotes the predicted vessel regions in the overlay images. In the error maps, red and blue indicate false-positive and false-negative pixels, respectively, whereas white and black represent correctly classified vessel and background pixels. Several baseline methods produce fragmented thin vessels or fail to recover peripheral branches. In contrast, the proposed method in subfigure~(j) preserves more continuous vascular structures and exhibits fewer visible false-positive and false-negative regions, particularly around fine vessel branches. These qualitative observations are consistent with the quantitative DRIVE results reported in Table~\ref{tab:quantitative}, supporting the effectiveness of the geometric constraints in improving boundary continuity and stabilizing zero-level-set evolution for elongated tubular structures.
\section{Conclusion}
This paper presented an image-conditioned Flow Matching framework for medical image segmentation that represented masks as SDFs and incorporated high-order geometric constraints into continuous transport. Biharmonic regularization promoted smooth and structurally coherent implicit representations, while the Eikonal constraint preserved the distance-field property and improved boundary stability and shape fidelity during ODE-based evolution.
Experiments on MoNuSeg, GlaS, and DRIVE demonstrated competitive region-overlap and boundary-accuracy results. The ablation and NFE studies further showed that the geometric constraints reduced performance fluctuations and enabled stable segmentation with few integration steps, providing a favorable balance between accuracy and inference efficiency.
Potential future directions included extending the framework to multi-class segmentation, diverse imaging modalities, and larger clinical datasets, as well as exploring its combination with topology-preserving objectives and lightweight ODE solvers.


\end{document}